\documentclass[10pt]{article}
\usepackage[latin9]{inputenc}
\usepackage{geometry}
\usepackage{float}
\usepackage{bm}
\usepackage{amsthm}
\usepackage{amsmath}
\usepackage{amssymb}
\usepackage{graphicx}

\makeatletter

\providecommand{\tabularnewline}{\\}
\floatstyle{ruled}
\newfloat{algorithm}{tbp}{loa}
\providecommand{\algorithmname}{Algorithm}
\floatname{algorithm}{\protect\algorithmname}

\theoremstyle{plain}
\newtheorem{thm}{\protect\theoremname}
  \theoremstyle{plain}
  \newtheorem{lem}[thm]{\protect\lemmaname}

\@ifundefined{date}{}{\date{}}
\usepackage{framed}
\usepackage{amssymb,amsfonts,amsmath}
\usepackage{amsfonts,amsmath}
\usepackage{graphicx}
\usepackage{subfigure,enumerate,bm,amsmath,amsthm,wrapfig,array,color,algorithmic}
\usepackage{graphics}
\usepackage{url}
\usepackage{multirow}
\usepackage{setspace}

\usepackage{times}

\newenvironment{assumption}[2][Assumption]{\begin{trivlist}
\item[\hskip \labelsep {\bfseries #1}\hskip \labelsep {\bfseries #2}]}{\end{trivlist}}

\makeatother

  \providecommand{\lemmaname}{Lemma}
\providecommand{\theoremname}{Theorem}

\begin{document}

\title{Revisiting Asynchronous Linear Solvers: Provable Convergence Rate
Through Randomization%
\thanks{An extended abstract of this work appears in the proceedings of the 28th IEEE International Parallel \& Distributed
Processing Symposium (IPDPS), 2014~\cite{ADG14_IPDPS}.%
}}

\author{Haim Avron \\
 IBM T.J. Watson Research Center\\
 haimav@us.ibm.com\\
 \and Alex Druinsky \\
 Lawrence Berkeley National Laboratory \\
adruinsky@lbl.gov \\
\and Anshul Gupta \\
IBM T.J. Watson Research Center \\
 anshul@us.ibm.com}
\maketitle
\begin{abstract}
Asynchronous methods for solving systems of linear equations have
been researched since Chazan and Miranker's pioneering 1969 paper
on chaotic relaxation. The underlying idea of asynchronous methods
is to avoid processor idle time by allowing the processors to continue
to make progress even if not all progress made by other processors
has been communicated to them.

Historically, the applicability of asynchronous methods for solving linear equations 
was limited to certain restricted classes of matrices, such as diagonally dominant
matrices. Furthermore, analysis of these methods
focused on proving convergence in the limit. Comparison of the asynchronous
convergence rate with its synchronous counterpart and its scaling
with the number of processors were seldom studied, and are still not
well understood. 

In this paper, we propose a randomized shared-memory asynchronous method for general
symmetric positive definite matrices. We rigorously analyze the convergence
rate and prove that it is linear, and is close to that of the method's
synchronous counterpart if the processor count is not excessive relative
to the size and sparsity of the matrix. We also present an algorithm
for unsymmetric systems and overdetermined least-squares. Our work
presents a significant improvement in the applicability of asynchronous linear solvers
as well as in their convergence analysis, and suggests
randomization as a key paradigm to serve as a foundation for asynchronous
methods.
\end{abstract}

\global\long\def\R{\mathbb{R}}

\global\long\def\e{{\mathbf{e}}}

\global\long\def\x{{\mathbf{x}}}

\global\long\def\d{{\mathbf{d}}}

\global\long\def\b{{\mathbf{b}}}

\global\long\def\y{{\mathbf{y}}}

\global\long\def\z{{\mathbf{z}}}

\global\long\def\rb{{\mathbf{r}}}

\global\long\def\mat#1{{\ensuremath{\bm{\mathrm{#1}}}}}

\global\long\def\matN{\ensuremath{{\bm{\mathrm{N}}}}}

\global\long\def\matA{\ensuremath{{\bm{\mathrm{A}}}}}

\global\long\def\matP{\ensuremath{{\bm{\mathrm{P}}}}}

\global\long\def\matU{\ensuremath{{\bm{\mathrm{U}}}}}

\global\long\def\matM{\ensuremath{{\bm{\mathrm{M}}}}}

\global\long\def\matR{\mat R}

\global\long\def\matS{\mat S}

\global\long\def\matY{\mat Y}

\global\long\def\matI{\mat I}

\global\long\def\matJ{\mat J}

\global\long\def\matZ{\mat Z}

\global\long\def\matL{\mat L}

\global\long\def\TNormS#1{\|#1\|_{2}^{2}}

\global\long\def\TNorm#1{\|#1\|_{2}}

\global\long\def\InfNorm#1{\|#1\|_{\infty}}

\global\long\def\FNorm#1{\|#1\|_{F}}

\global\long\def\UNorm#1{\|#1\|_{\matU}}

\global\long\def\UNormS#1{\|#1\|_{\matU}^{2}}

\global\long\def\UINormS#1{\|#1\|_{\matU^{-1}}^{2}}

\global\long\def\ANorm#1{\|#1\|_{\matA}}

\global\long\def\BNorm#1{\|#1\|_{\mat B}}

\global\long\def\ANormS#1{\|#1\|_{\matA}^{2}}

\global\long\def\XNormS#1{\|#1\|_{\mat X}^{2}}

\global\long\def\AINormS#1{\|#1\|_{\matA^{-1}}^{2}}

\global\long\def\T{\textsc{T}}

\global\long\def\Expect#1{{\mathbb{E}}\left[#1\right]}

\global\long\def\ExpectC#1#2{{\mathbb{E}}_{#1}\left[#2\right]}

\global\long\def\dotprod#1#2#3{(#1,#2)_{#3}}

\global\long\def\dotprodsqr#1#2#3{(#1,#2)_{#3}^{2}}

\global\long\def\Trace#1{{\bf Tr}\left(#1\right)}

\global\long\def\nnz#1{{\bf nnz}\left(#1\right)}

\section{Introduction}

It has long been recognized that high global synchronization costs
will eventually limit the scalability of iterative solvers. So early
on, starting with the pioneering work of Chazan and Miranker on \emph{chaotic
relaxation} in 1969~\cite{CM69} (see review by Frommer and Szyld~\cite{FS00}),
asynchronous methods have been researched and deployed. These methods
avoid synchronization points and their associated costs by allowing
processors to continue to work even if not all progress made by other
processors has been communicated to them.

While asynchronous methods were successfully applied to many numerical
problems~\cite{FS00}, interest in them dwindled over the years.
One important reason is that until recently, concurrency was not large
enough to warrant the use of asynchronous methods, as asynchronous
methods typically require more computation when compared to their
synchronous counterparts. Other reasons are related to the limits
of existing theory on asynchronous methods. Historically, the 
applicability of asynchronous methods for solving linear equations 
was limited to restricted classes of matrices, such as diagonally dominant
matrices. This had a substantially negative impact on the relevance and 
interest in asynchronous methods, as most of the matrices arising 
in applications did not posses the required attributes.  
Furthermore, analysis of asynchronous methods for solving linear equations 
focused on proving
convergence in the limit. How the rate of convergence compares to
the rate of convergence of the synchronous counterparts, and how this
rate scales when the number of processors increases, was seldom studied
and is still not well understood. It was observed experimentally that
asynchronous methods can sometimes be substantially slower than their
synchronous counterparts~\cite{BBDH12}. 

Today, as we push towards extreme scale systems, asynchronous algorithms
are becoming more and more attractive. In addition to the high synchronization
costs due to massive parallelism, other hardware issues make asynchronous
methods attractive as well. Current hardware trends suggest that software
running on extreme-scale parallel platforms will be expected to encounter
and be resilient to nondeterministic behavior from the underlying
hardware. Asynchronous methods are inherently well-suited to meet
this challenge. On the other hand, it is also clear that a paradigm
shift regarding the way asynchronous methods are designed and analyzed
must be made, if such methods are to be deployed. To that end, this
paper makes three significant contributions. It presents an asynchronous
solver with randomization as a key algorithmic component, a rigorous
analysis that affirms the role of randomization as an effective tool
for improving asynchronous solvers, and an analytical methodology
for asynchronous linear solvers based on a realistic bounded-delay
model.

Specifically, we present a new asynchronous shared-memory parallel
solver for symmetric positive definite matrices with a provable linear
convergence rate under a mostly asynchronous computational model that
assumes bounded delays. A key component of our algorithm is randomization,
which allows the processors to make progress independently with only
a small probability of interfering with each other. Our analysis shows
a convergence rate that is linear in the condition number of the matrix,
and depends on the number of processors and the degree to which the
matrix is sparse. A slightly better bound is achieved if we occasionally
synchronize the processors. In either case, as long as the number
of processors is not too large relative to the size and sparsity of
the matrix, the convergence rate is close to that of the synchronous
counterpart. Unlike in general asynchronous methods, the convergence rate
does not depend on numerical classification of the matrix (e.g., diagonal
dominance). In particular, our method will converge for essentially
\emph{any }large sparse symmetric positive definite matrix as long
as not too many processors are used. We also present an algorithm
for unsymmetric systems and overdetermined least-squares. 

Our method and its analysis do have some limitations. Adapting the algorithm
to the distributed memory setting is not straightforward. Our algorithm allows 
each processor to update all the entries of the solution vector, but in 
a distributed memory setting it is desirable that each processor
owns and be the sole updater of only a subset of the entries. To allow this, a more limited 
form of randomization should be used, and this is not explored in the paper.
Our algorithm also tends to generate much more cache misses than classical 
asynchronous methods for structured matrices. Again, it may be possible to circumvent
this using a more restricted form of randomization. More importantly, our
algorithm's convergence is inherently slower than that of Krylov-subspace
methods, which is a feature of the underlying synchronous algorithm. For this
reason, the algorithm is most suitable when only moderate accuracy is sought,
either when we require low accuracy in the ultimate solution or when we use the
algorithm as a preconditioner in a flexible Krylov method.
Our algorithm relies on some assumptions that are hard
to enforce in practice, and the convergence results have a parameter
which is hard to quantify better than by just a very rough upper bound.
Finally, we remark that bounds tend to be rather pessimistic (this is true 
also for the synchronous algorithm that is the basis for our algorithm).

Nevertheless, even with these limitations we believe our work
presents a significant improvement in the applicability of asynchronous 
linear solvers, as well as in the convergence analysis, and suggests
randomization as a key paradigm to serve as a foundation for asynchronous
methods.

While the primary aim of this paper is to present analytical results,
we also include some experimental results. With our implementation,
we are able to demonstrate that the proposed method can be attractive
for certain types of linear systems even in the absence of massive
parallelism. Previous asynchronous methods, as well as ours, are
based on basic iteration (e.g., Gauss-Seidel). Those are known to convergence very slowly
in the long run when compared with Krylov subspace methods. However,
big data applications typically require very low accuracy, so they
are better served using basic iterations as these tend to initially
converge very quickly and scale better. Our experiments show that
for a linear system arising from analysis of social media data, our
proposed algorithm scales well, pays very little to no penalty for
asynchronicity, and overall seems to present the best choice for solving
the said linear system to the required accuracy.

We review related work in Section~\ref{sec:related}. Essential background
on Randomized Gauss-Seidel is given in Section~\ref{sec:rgs}. In
Section~\ref{sec:ags} we propose two asynchronous models for executing
Randomized Gauss-Seidel: one assumes that consistent reads have been
enforced, another does not. Section \ref{sec:with} analyzes the convergence
when the consistent read assumption is enforced. Section \ref{sec:step}
shows that convergence can be improved if we control the step-size.
In Section \ref{sec:without}, we analyze convergence rate when we
allow inconsistent reads. We briefly discuss unsymmetric systems and
overdetermined least-squares in Section~\ref{sec:LS}. Section~\ref{sec:experiments}
presents experimental results. Finally, in Section~\ref{sec:conclusions}
we make some concluding remarks and discuss future work.

\subsection*{Setup and Notation}

Most of this paper is concerned with solving the linear equation $\matA\x=\b$
where $\matA\in\R^{n\times n}$ is a symmetric positive definite matrix,
and $\b\in\mathbb{R}^{n}$. For simplicity we assume that $\matA$
has a unit diagonal. This is easily accomplished using re-scaling.
Our results can be easily generalized to allow an arbitrary diagonal,
but making this assumption helps keep the presentation and notation
more manageable. We denote the exact solution to this equation by
$\x^{\star}$, i.e. $\x^{\star}=\matA^{-1}\b$. We denote the largest
eigenvalue of $\matA$ by $\lambda_{\max}$, and the smallest eigenvalue
by $\lambda_{\min}$. The condition number of $\matA$, which is equal
to $\lambda_{\max}/\lambda_{\min}$, is denoted by $\kappa$.

We are predominantly interested in the case where $\matA$ is sparse
and very large, and the number of non-zeros in each row is between
$C_{1}$ and $C_{2}\ll n$ with a small ratio between $C_{2}$ and
$C_{1}$.\emph{ }This scenario frequently occurs in many scientific
computing applications. Throughout the paper we refer to this scenario
as the \emph{reference scenario.} We state and prove more general
results; we do not use the properties of the reference scenario in
the proofs. The reference scenario is mainly useful for the interpretation
of the practical implications of the results. Note that in the reference
scenario we have $\lambda_{\max}\leq C_{2}\ll n$, as $\matA$ has
a unit-diagonal (so off-diagonal entries must be smaller than or equal
to one).

We use $\dotprod{\cdot}{\cdot}{\matA}$ to denote the $\matA$ inner
product. That is, $\dotprod{\x}{\y}{\matA}\equiv\y^{\T}\matA\x$ where
$\x,\y\in\R^{n}$. The fact that $\matA$ is a symmetric positive
definite matrix guarantees that $\dotprod{\cdot}{\cdot}{\matA}$ is
an inner product. The $\matA$-norm is defined by $\ANorm{\x}\equiv\sqrt{\dotprod{\x}{\x}{\matA}}.$
We use $\e^{(1)},\e^{(2)},\dots,\e^{(n)}$ to denote the $n$-dimensional
identity vectors (i.e. $\e^{(i)}$ is one at position $i$ and zero
elsewhere). $\matA_{i}$ denotes row $i$ of $\matA$, and $\matA_{ij}$
denotes the $i,j$ entry of $\matA$. We will generally use subscript
indices on vectors for iteration counters. The notation $(\x)_{i}$
denotes the $i$th entry of $\x$.

Throughout the paper we describe algorithms that generate a series
of approximations to $\x^{\star}$, denoted by $\x_{0},\x_{1},\dots$
(subscript index is the iteration counter), which are actually random
vectors. We denote the expected squared $\matA$-norm of the error
of $\x_{m}$ by $E_{m}$. That is,
\[
E_{m}\equiv\Expect{\ANormS{\x_{m}-\x^{\star}}}\,.
\]

\section{\label{sec:related}Related Work}

Asynchronous methods were first suggested by Chazan and Miranker~\cite{CM69}
in their pioneering paper on chaotic relaxation. The theory and application
of asynchronous iterations has since been studied and used by many
authors. Noteworthy is the seminal text by Bertsekas and Tsitsiklis~\cite{BT89}.
A more recent review is by Frommer and Szyld~\cite{FS00}.

Historically, work on asynchronous methods focused on proving that
the methods converge in the limit, and not on convergence rate analysis.
In particular, the relation to the convergence rate of synchronous
counterparts, and the scaling of these methods, were seldom studied.
We are aware of only two exceptions of work published before 2011,
but the results are unsatisfactory. Baudet~\cite{Baudet78} generalizes 
the Chazan-Miranker result to nonlinear mappings using the same model of asynchronism. 
For the linear case, he does not extend the class of systems that can be 
solved in this model (the Chazan-Miranker result is an 
if and only if result, so there is no possibility for improvement 
in that model in that respect), but he does show that given a 
trace of an asynchronous iteration,
if in that trace the maximum delay on purging old information and
the rate of updating all components is bounded, then convergence is
linear with the rate of convergence divided by the sum of these two
bounds. The author does not analyze and show rigorous convergence
rates for a concrete algorithm for solving a linear system, nor compare the convergence rate of
some specific asynchronous algorithm to its synchronous counterpart.
Bertsekas and Tsitsiklis~\cite[Section 7.2, Exercise 1.2]{BT89}
prove a linear convergence rate of certain asynchronous iterations
for some classes of matrices (like weakly diagonally dominant matrices),
but analyze how the rate of convergence depends on the measure of
asynchronism only under very restrictive conditions and in a hard
to interpret manner~\cite[Section 6.3.5]{BT89}.

Following the influential work of Niu et al.~\cite{RRWN11}, recent
work has focused on analyzing the rate of convergence. When discussing
these recent results it is important to distinguish between methods
that assume a \emph{consistent read }model, and models that allow
\emph{inconsistent reads}. Informally (we give a formal definition
in Section~\ref{sec:ags}), the consistent read model assumes that
the part of the state (i.e., iteration vector) that the algorithm reads
in order to update the state is consistent with a state that existed
in shared memory at some point in time. Without special provisions,
which might be computationally expensive, the consistent read assumption
is somewhat unsatisfactory, although the analysis of a consistent read
model is not without merit, as we explain in Section~\ref{sec:ags}.

The basic model proposed by Niu et al.~\cite{RRWN11} is as follows.
There is an iteration vector $\x$ that is stored in shared-memory.
All processors share this memory and update it in an asynchronous,
uncoordinated fashion, without any form of locking. This implies that
the version of $\x$ that is used by a processor to update $\x$ is
not the same as the version on which this update is applied, as $\x$
has possibly been updated by other processors in the interim. It is
however assumed that write operations are atomic, and that there is
a bound $\tau$ on how many updates are missed. This model matches
modern multicore architectures well. Under an additional assumption
of consistent reads, Niu et al. establish a sub-linear convergence
rate of asynchronous stochastic gradient descent. Our model follows
the one proposed by Niu et al.~\cite{RRWN11}, although we analyze
convergence both in the consistent read and inconsistent read model.
We also establish a linear convergence rate, unlike the sublinear
rate for stochastic gradient descent, with a better dependence on
$\tau$.

Liu et al.~\cite{LWS14} propose an Asynchronous Randomized Kaczmarz
algorithm for solving consistent square and overdetermined linear
systems. They use the same model as Niu et al.\ do, and assume consistent
reads as well. A linear convergence rate is established. An extension
to inconsistent systems is discussed as well.

Later, Liu et al.~\cite{LWRBS14} develop an asynchronous stochastic
coordinate descent algorithm. Again, they use the same model as Niu et
al., and continue to assume consistent reads. Furthermore, they assume
that the radius of the iterate set is bounded, which is a condition
that might be hard to enforce in an asynchronous linear solver. They prove a
sublinear ($1/m$) convergence on general convex functions and a linear
convergence rate on functions that satisfy an ``essential strong
convexity'' property.

More recently, Liu et al.~\cite{LW14} suggest an asynchronous stochastic
proximal coordinate-descent algorithm for composite objective functions.
They allow inconsistent reads, and prove linear convergence for optimally
strongly convex functions, and a sublinear rate for general smooth
convex functions.

Our algorithm is closely related to the stochastic coordinate descent
algorithm in the sense that in essence our algorithm is an asynchronous
stochastic coordinate descent method applied to the strongly convex
quadratic optimization problem $\min_{\x}\frac{1}{2}\x^{\T}\matA\x-\b^{\T}\x$.
However, our results are much more tuned and interpretable to the
problem we consider ($\mat A\x=\b$) than the convergence result
for general strongly convex functions.

Hook and Dingle~\cite{HD13} analyze the convergence of the Jacobi
iteration when it is executed asynchronously on a distributed memory
machine. They prove upper and lower bounds on the convergence rate
of the iteration that are formulated in terms of the spectral radius
of $\mat A$ and two parameters of the asynchronous execution dynamics. Their
results indicate when convergence takes place and how fast it is,
even without the help of randomization. The dependence of the bounds
on parameters of the execution dynamics makes the convergence guarantee
hard to interpret. Nevertheless, the results show that performance
can suffer if an entry of the iterate is repeatedly updated using
stale data because of a slow communication link or fails to be updated
at all because of a slow processor. This indicates the potential of
using randomization to obtain robust performance in the face of such
single-point-of-failure vulnerabilities.

Unrelated to the previous, we also note Freris and Zouzias's~\cite{FZ12} work on using an asynchronous
variant of randomized Kaczmarz~\cite{SV09} to synchronize clocks
in a wireless network. They analyze the convergence rate in a semi-asynchronous
model that is suitable for wireless networks, but not for shared-memory
numerical computations.

\section{\label{sec:rgs}Randomized Gauss-Seidel}

Our asynchronous algorithm is based on the randomized variant of the
Gauss-Seidel iteration, originally proposed by Leventhal and Lewis~\cite{LL10}.
We actually use a slight modification due to Griebel and Oswald~\cite{GO12}
that introduces a step-size (akin to under- and over-relaxation).  The goal of this section
is to describe and review the basic properties of the randomized Gauss-Seidel
iteration.

Consider the following iteration applied to some arbitrary initial
vector $\x_{0}\in\R^{n}$, and a series of direction vectors $\d_{0},\d_{1},\dots$:
\[
\begin{alignedat}{1}\rb_{j} & =\,\b-\matA\x_{j}\\
\gamma_{j} & =\,\d_{j}^{\T}\rb_{j}\\
\x_{j+1} & =\,\x_{j}+\beta\gamma_{j}\d_{j}\,,
\end{alignedat}
\]
where $0<\beta<2$. In terms of the analysis it is more convenient to write the iteration
in the following equivalent form:
\begin{equation}
\begin{alignedat}{1}\gamma_{j} & =\,\dotprod{\x^{\star}-\x_{j}}{\d_{j}}{\matA}\\
\x_{j+1} & =\,\x_{j}+\beta\gamma_{j}\d_{j}.
\end{alignedat}
\label{eq:basic-iteration}
\end{equation}
Both iterations are listed to show that even though the unknown $\x^{\star}$
appears in~\eqref{eq:basic-iteration}, the iteration is computable.

In~\eqref{eq:basic-iteration} the scalars $\gamma_{0},\gamma_{1},\dots$
are selected so as to minimize $\ANorm{\x^{\star}-\x_{j+1}}$ when
$\x_{j+1}$ is obtained from $\x_{j}$ by taking a step in the direction
$\d_{j}$ with $\beta=1$. There are quite a few ways to set $\d_{0},\d_{1},\dots$.
Each is associated with a different per-iteration cost, and different
convergence properties. One well known method is setting $\d_{i}=\e^{((i\,\mod\, n)+1)}$.
In that case, every $n$ iterations corresponds to a single iteration
of Gauss-Seidel (recall that we assume that the matrix has unit diagonal).

Leventhal and Lewis suggested using random directions instead of deterministic
ones: $\d_{0},\d_{1},\dots$ are i.i.d. random vectors, taking $\e^{(1)},\dots,\e^{(n)}$
with equal probability%
\footnote{Leventhal and Lewis consider the more general setting where $\matA$
does not have unit diagonal. For that case, they analyze non-uniform
probabilities. When the matrix has unit diagonal, their algorithm
and the convergence analysis reduces to the ones stated here.%
}. For this distribution of direction vectors, Griebel and Oswald~\cite{GO12}
prove the following
bound on the expected error in the $\matA$-norm (the case of $\beta=1$ was
analyzed by Leventhal and Lewis~\cite{LL10}):
\begin{equation}
E_{m}\leq\left(1-\frac{\beta(2-\beta)\lambda_{\min}}{n}\right)^{m}\ANormS{\x_{0}-\x^{\star}}\,.\label{eq:rgs-bound}
\end{equation}
So, the randomized Gauss-Seidel iteration converges in expectation
at a linear rate%
\footnote{Some care should be employed with terminology. Some mathematicians
or computer scientists might say this is an \emph{exponential }or
\emph{geometric convergence rate. }However, numerical analysts refer
to this rate as \emph{linear, }as it is linear in $O(\log(\epsilon)$)
where $\epsilon$ is the desired reduction factor of the error. %
}. Markov's inequality now implies that given $\epsilon>0$ and $\delta\in(0,1)$,
for
\[
m\geq\frac{n}{\beta(2-\beta)\lambda_{\min}}\ln\left(\frac{1}{\delta\epsilon^{2}}\right)
\]
we have
\[
\Pr(\ANorm{\x_{m}-\x^{\star}}\geq\epsilon\ANorm{\x_{0}-\x^{\star}})\leq\delta\,.
\]
If we could compute $\ANorm{\x_{m}-\x^{\star}}$, this will imply
a randomized algorithm whose probabilistic guarantees are only on
the running time, and not on the quality of approximation. In practice,
we can check the residual $\TNorm{\b-\mat A\x_{m}}$, as is typically
done in iterative methods. Similar transformations can be done to
other bounds throughout this paper. These transformations are rather
technical, so we omit them. Note that the expected cost per iteration
of randomized Gauss-Seidel is $\Theta(\nnz{\matA}/n)$, so $n$ iterations
(which we refer to as a \emph{sweep) }are about as costly as a single
Gauss-Seidel iteration.

The proof of \eqref{eq:rgs-bound} relies on the following lemma,
which we use extensively in our analysis as well. The upper bound
in the lemma was not proven by Leventhal and Lewis~\cite{LL10},
but it can be proved using the same technique they used to prove the
lower bound. For completeness we include a proof.
\begin{lem}
\label{lem:key}Let $\d$ be a random vector taking $\e^{(1)},\dots,\e^{(n)}$
with equal probability. Suppose that $\x$ and $\d$ are independent.
Then,
\[
\frac{\lambda_{\min}}{n}\Expect{\ANormS{\x-\x^{\star}}}\leq\Expect{\dotprodsqr{\x-\x^{\star}}{\d}{\matA}}\leq\frac{\lambda_{\max}}{n}\Expect{\ANormS{\x-\x^{\star}}}\,.
\]
\end{lem}
\begin{proof}
Let $\mat B$ be the unique symmetric positive matrix such that $\matA=\mat B^{2}$.
We have
\begin{eqnarray*}
\Expect{\dotprodsqr{\x-\x^{\star}}{\d}{\matA}} & = & \Expect{\Expect{\dotprodsqr{\x-\x^{\star}}{\d}{\matA}}\middle|\,\x}\\
 & = & \Expect{\frac{1}{n}\sum_{i=1}^{n}\dotprodsqr{\x-\x^{\star}}{\e_{i}}{\matA}}\\
 & = & \frac{1}{n}\Expect{\TNormS{\matA(\x-\x^{\star})}}\\
 & = & \frac{1}{n}\Expect{(\x-\x^{\star})^{\T}\matA^{2}(\x-\x^{\star})}\\
 & = & \frac{1}{n}\Expect{\frac{(\x-\x^{\star})^{\T}\mat B\matA\mat B(\x-\x^{\star})}{(\x-\x^{\star})^{\T}\mat B\mat B(\x-\x^{\star})}\cdot(\x-\x^{\star})^{\T}\mat B\mat B(\x-\x^{\star})}\\
 & = & \frac{1}{n}\Expect{\frac{(\x-\x^{\star})^{\T}\mat B\matA\mat B(\x-\x^{\star})}{(\x-\x^{\star})^{\T}\mat B\mat B(\x-\x^{\star})}\cdot\ANormS{\x-\x^{\star}}}\,.
\end{eqnarray*}
According to the Courant-Fischer theorem, for every vector $\y\neq0$
we have
\[
\lambda_{\min}\leq\frac{\y^{\T}\matA\y}{\y^{\T}\y}\leq\lambda_{\max}\,.
\]
Applying the last inequality to the previous equality with $\y=\mat B(\x-\x^{\star})$
completes the proof.
\end{proof}

\subsubsection*{Non-Unit Diagonal}

We now explain why there is no loss in generality in assuming that
$\matA$ has unit diagonal.

Suppose that $\mat B$ does not have unit diagonal. Consider the following
more general Randomized Gauss-Seidel iteration (also due to Leventhal
and Lewis~\cite{LL10}):
\begin{equation}
\begin{alignedat}{1}\tilde{\gamma}_{j} & =\,\frac{\dotprod{\y^{\star}-\y_{j}}{\d_{j}}{\mat B}}{\dotprod{\d_{j}}{\d_{j}}{\mat B}}\\
\y_{j+1} & =\,\y_{j}+\beta\tilde{\gamma}_{j}\d_{j}.
\end{alignedat}
\label{eq:basic-iteration-nonunit}
\end{equation}
where $\y^{\star}$ is the solution to $\mat B\y=\z$, and $\d_{0},\d_{1},\dots$
are i.i.d. random vectors taking $\e^{(1)},\dots,\e^{(n)}$ with equal
probability. Let $\mat D$ be the diagonal matrix such $\mat A=\mat D\mat B\mat D$
has unit diagonal, and consider the unit-diagonal Randomized Gauss-Seidel
iteration~\eqref{eq:basic-iteration} for the linear system $\matA\x=\mat D\z$
using the same direction vectors $\d_{0},\d_{1},\dots$. It is not
hard to verify that $\y_{j}=\mat D\x_{j}$ and that $\ANorm{\x_{j}-\x^{\star}}=\BNorm{\y_{j}-\y^{\star}}$.
Therefore, analyzing the unit-diagonal scenario is sufficient.

\section{\label{sec:ags}Asynchronous Randomized Gauss-Seidel (AsyRGS)}

\begin{algorithm}
\begin{algorithmic}[1]

\STATE \textbf{Input:} $\matA\in\R^{n\times n}$, $\b\in\R^{n}$, (pointer to) 
vector $\x$ (initial approximation and algorithm output),
$\beta\in(0,2)$.

\STATE

\LOOP

\STATE Pick a random $r$ uniformly over $\{1,\dots,n\}$

\STATE Read the entries of $\x$ corresponding to non-zero entries
in $\matA_{r}$

\STATE Using these entries, compute $\gamma\gets(\b)_{r}-\matA_{r}\x$

\STATE Update: $(\x)_{r}\gets(\x)_{r}+\beta \gamma$

\ENDLOOP

\end{algorithmic}

\protect\caption{\label{alg:rgs}Randomized Gauss-Seidel}
\end{algorithm}

Algorithm~\ref{alg:rgs} contains a pseudo-code description of randomized
Gauss-Seidel in which we made the read and update operations explicit.
This obviously entails some details that are, in a sense, implementation
specific. There are implementations of the randomized Gauss-Seidel
iteration which do not match the description in Algorithm~\ref{alg:rgs}.

Consider a shared memory model with $P$ processors. Each processor
follows Algorithm~\ref{alg:rgs} using the same $\x$, i.e. all processors
read and update the same $\x$ stored in a shared memory. The processors
do not explicitly coordinate or synchronize their iterations. We do,
however, impose assumptions, some of which may require enforcement
 in an actual implementation. The first assumption is rather simple:
the update operation in each iteration is atomic.

\begin{assumption}{A-1} (Atomic Write). The update operation in line
7 is atomic.\end{assumption}

The update operation operates on a single coordinate in $\x$. For
single- or double-precision floating point numbers, updates
of the form used in line 7 have hardware support on many modern processors
(e.g. compare-and-exchange on recent Intel processors).

If atomic write is enforced, then for the sake of the analysis we
can impose an order $\x_{0},\x_{1},\x_{2},\dots$ on the values that
$\x$ takes during the computation. Here $\x_{j}$ denotes the value
of $\x$ after $j$ updates have been applied (breaking ties in an
arbitrary manner).

We now turn our attention to the read operation in line 5. Here we
consider two possible models. In the first model, we assume the following
consistent read assumption is enforced.

\begin{assumption}{A-2} (Consistent Read). The values of the entries
of $\x$ read in line 5 appeared together in $\x$ at some time before
the update operation (line 7) is executed. \end{assumption}

Note that Assumption~A-2 does not necessarily imply that none of
the entries read during the execution of line 5 are modified while
that line is being executed; this is only one way of enforcing this
assumption. More formally, if $R$ denotes the set of entries read
in the execution of line 5 for a particular execution of the iteration,
and $M$ denotes the set of entries modified during the execution
of line 5 in that iteration, then Assumption~A-2 holds for that iteration
if $R\cap M=\emptyset$. However, this is only a sufficient condition,
not a necessary one.

With consistent read, we can denote by $k(j)\leq j$ the maximum iteration
index such $\x_{k(j)}$ is equal to the values read on line 5, on
the indices read during the execution of line 5. The existence of
such a $k(j)$ is guaranteed by Assumption~A-2 (since all writes
are atomic, all time intervals correspond to some iteration index).
The iteration can then be written:
\begin{equation}
\begin{alignedat}{1}\gamma_{j} & =\,\dotprod{\x^{\star}-\x_{k(j)}}{\d_{j}}{\matA}\\
\x_{j+1} & =\,\x_{j}+\beta\gamma_{j}\d_{j}\,.
\end{alignedat}
\label{eq:a-iteration}
\end{equation}

We also consider a model where we allow inconsistent reads.
Since every iteration changes a single coordinate, and we require all writes
to be atomic, the value of $\x$ read in line 5 is the result of a
subset of the updates that occurred before the write operation in
line 7 is executed. Let us denote by $K(j)\subseteq\{0,1,\dots,j-1\}$
a maximal set of updates consistent with the computation of $\gamma$
in iteration $j$. In other words, an index $i\leq j$ is in $K(j)$
if either it updates an entry of $\x$ not read for computing $\gamma_{j}$,
or it updates an entry and the update was applied before that entry
was read. The entries read are consistent with the vector
\[
\x_{K(j)}=\x_{0}+\sum_{i\in K(j)}\beta\gamma_{i}\d_{i}\,.
\]
Note that the $\x_{K(j)}$ might have never existed in memory during
the execution of the algorithm. Nevertheless, the iteration can be
written as
\begin{equation}
\begin{alignedat}{1}\gamma_{j} & =\,\dotprod{\x^{\star}-\x_{K(j)}}{\d_{j}}{\matA}\\
\x_{j+1} & =\,\x_{j}+\beta\gamma_{j}\d_{j}\,.
\end{alignedat}
\label{eq:uc-iteration}
\end{equation}

Obviously, enforcing consistent reads involves some overhead.
In the analysis we consider both models, as the  bounds for the
inconsistent read model are not as good as the
ones obtained when assuming consistent reads. There is clearly a trade-off
here, which we present but do not attempt to quantify. It is a complex
trade-off that depends on many factors, including possible hardware
features like transactional memory that may enable efficient enforcement
of consistent reads.

More importantly, in many cases even\emph{ without }any special provisions,
the probability of an inconsistent read in an iteration is extremely
small, so much that we do not expect it to happen much (or at all)
in a normal execution of the algorithm. For the definition of consistent
read to be violated in a certain iteration there must be two distinct
indices $l$ and $c$ for which all of the following conditions are
met:
\begin{enumerate}
\item $\matA_{rc}\neq0$ and $\matA_{rl}\neq0$ ($r$ is the index picked
in line 4).
\item Both $(\x)_{c}$ and $(\x)_{l}$ are modified by other processors
during the execution of line 5 (the read operation).
\item $(\x)_{c}$ is read before $(\x)_{l}$.
\item $(\x)_{c}$ is modified (by another processor) \emph{after} it is
read, and $(\x)_{l}$ is modified (by another processor) \emph{before
}it is read.
\end{enumerate}
Having all these condition occur at the same time is rather rare.
In fact, just having the first two occur is rather rare in the reference
scenario. The reason is that each iteration reads at most $C_{2}\ll n$
entries. Suppose there are $u$ updates while reading those entries.
Each such update affects a single random entry. Therefore, the probability
that it will update one of the $C_{2}$ entries being read is at most
$C_{2}/n$. The probability of getting two such updates is bounded
by the probability of getting at least two in a binomial distribution
with $u$ experiments and probability $C_{2}/n$. Unless $u$ is very
large, this is an extremely small probability (since $C_{2}/n$ is
tiny).

The discussion above suggests that in many cases the bound we obtain
for the consistent read will be rather descriptive even if no special
provisions are taken to enforce the consistent read assumption. That
is, we expect the actual behavior to be somewhere between the bound
for the consistent read and that for the inconsistent
read, but closer to the one for consistent read.

We are mainly interested in algorithms with provable convergence rate.
In a totally asynchronous model with arbitrary delays, there can also
be an arbitrary delay in convergence. Therefore, we assume that asynchronism
is bounded in the sense that delays are bounded.

\begin{assumption}{A-3} (Bounded Asynchronism). There is a constant
$\tau$ (measure of asynchronism) such that all updates that are older
than $\tau$ iterations participate in the computation of iteration
$j$, for all iterations $j=1,2,\dots$. \end{assumption}

In the consistent read model, this assumption translates to requiring
that

\begin{equation}
j-\tau\leq k(j)\leq j\,.\label{eq:delay}
\end{equation}
In the inconsistent read model, this assumption translates to requiring
that
\begin{equation}
\{0,1,\dots,\max\{0,j-\tau-1\}\}\subseteq K(j)\,.\label{eq:delay-K}
\end{equation}
Since the running time of an iteration is proportional to the number
of non-zeros in the row, a reasonable upper bound on $\tau$ is $c\cdot C_{2}\cdot P/C_{1}$
for some small constant $c$. However, this is probably a pessimistic
upper bound, and in general  when the variance in the number of non-zeros
per row is not too large relative to the mean, we expect $\tau$ to
be of order of $P$. Regardless, it is clear that in the reference
scenario $\tau=O(P)$ (recall that we assume that $C_{2}/C_{1}$ is
a small constant).

We now discuss the relation between $k(0),k(1),\dots$ or $K(0),K(1),\dots$
and the random variables $\d_{0},\d_{1},\dots$. If we inspect the
pseudo-code of Algorithm~\ref{alg:rgs} closely we will realize that
$k(j)$ or $K(j)$ (depending on the model) depend on the random choices
$\d_{0},\d_{1},\dots,\d_{j-1}$ made before the write operation, and
more crucially on the random choice $\d_{j}$. The reason is that
on line 5 we read only the relevant entries of $\x$, so only a small
set of updates can be considered for inclusion. The set of relevant
entries is determined by the selection of $\d_{j}$. However, a completely
adversarial model which allows dependence of $k(j)$ (or $K(j)$)
on $\d_{0},\d_{1},\dots,\d_{j}$ (for $j=1,2,\dots)$ and analyzes
the worst-case behavior is not likely to be very faithful to the actual
behavior of the algorithm. Therefore, we assume the delays are independent
of the random choices, but allow them to be arbitrary (as long as
the bounded asynchronism assumption holds). We acknowledge that this assumption
cannot be enforced without paying a significant penalty in terms of iteration costs
(e.g., the assumption is satisfied if the algorithm reads all the entries of $\x$ in each step)
 
\begin{assumption}{A-4} (Independent Delays). We allow an arbitrary
set of delays that satisfy~\eqref{eq:delay} or~\eqref{eq:delay-K}
(depending on the context), but they do not depend on the random choices
$\d_{0},\d_{1},\dots$. \end{assumption}

Finally, we remark on the role of the various assumptions in the analysis.
Assumptions A-1 and A-2 allow us to write well defined iterations
(iterations~\eqref{eq:a-iteration} and \eqref{eq:uc-iteration})
that can be analyzed mathematically. Assumption A-3 allows us to bound
the number of elements in $\x_{j}-\x_{k(j)}$ (or $\x_{j}-\x_{K(j)}$
) that are non-zero by $\tau$, therefore implying that $\dotprod{\x^{\star}-\x_{k(j)}}{\d_{j}}{\matA}=\dotprod{\x^{\star}-\x_{j}}{\d_{j}}{\matA}$
(respectively, $\dotprod{\x^{\star}-\x_{K(j)}}{\d_{j}}{\matA}=\dotprod{\x^{\star}-\x_{j}}{\d_{j}}{\matA}$)
with high probability (as long as $\tau$ is small relative to $n$),
which implies that the computed step-size is correct in most iterations.
Assumption A-4 allows us to treat $k(j)$ and $K(j)$, in the proofs,
as deterministic even though they might be random. Since $k(0),k(1),\dots$
(or $K(0),K(1),\dots$) do not depend on $\d_{0},\d_{1},\dots$ we can
condition on choices for which equation~\eqref{eq:delay} (respectively,
equation~\eqref{eq:delay-K}) holds, and any bound that does not
depend on $k(0),k(1),\dots$ (respectively, $K(0),K(1),\dots$) will
hold for random ones as well.

For clarity, we now detail explicitly the two models we analyze in this paper.

\paragraph{Consistent Read Model.}
Algorithm~\ref{alg:rgs} executed on all processors using the same $\x$ (i.e., all processors
read and update the same $\x$ stored in a shared memory) with all four assumptions (A-1 to A-4). The
governing iteration is:
\begin{equation}
\begin{alignedat}{1}
j-\tau& \leq k(j) \leq j\\ 
\d_{j} & \sim U(\e^{(1)},\dots,\e^{(n)} )\\
\gamma_{j} & =\,\dotprod{\x^{\star}-\x_{k(j)}}{\d_{j}}{\matA}\\
\x_{j+1} & =\,\x_{j}+\beta\gamma_{j}\d_{j}~\,
\end{alignedat}
\label{eq:full-a-iteration}
\end{equation}
with the additional assumptions that $\d_0,\d_1,\dots$ are i.i.d, and 
that $k(0), k(1),\dots$ do not depend on the random choices $\d_{0},\d_{1},\dots$.
In the above, $U(\e^{(1)},\dots,\e^{(n)})$ denotes a uniform distribution on 
the $n$-dimensional identity vectors.

\paragraph{Inconsistent Read Model.}
Algorithm~\ref{alg:rgs} executed on all processors using the same $\x$ (i.e., all processors
read and update the same $\x$ stored in a shared memory) with assumptions A-1, A-3 and A-4. The
governing iteration is:
\begin{equation}
\begin{alignedat}{1}
\{0, & \dots,\max\{0,j-\tau-1\}\} \subseteq K(j) \subseteq \{0,\dots,j\} \\
\d_{j} & \sim U(\e^{(1)},\dots,\e^{(n)} )\\
\gamma_{j} & =\,\dotprod{\x^{\star}-\x_{K(j)}}{\d_{j}}{\matA}\\
\x_{j+1} & =\,\x_{j}+\beta\gamma_{j}\d_{j}~\,
\end{alignedat}
\label{eq:full-inc-iteration}
\end{equation}
with the additional assumptions that $\d_0,\d_1,\dots$ are i.i.d, and 
that $K(0), K(1),\dots$ do not depend on the random choices $\d_{0},\d_{1},\dots$.

\section{\label{sec:with}Convergence Bound with Consistent Read and Unit Step-size ($\beta=1$)}

In this section, we analyze the iteration under the consistent read model, i.e. 
iteration~\eqref{eq:full-a-iteration}. For the moment, we consider only unit step-size ($\beta=1$).
\begin{thm}
\label{thm:main}Consider iteration 
\eqref{eq:full-a-iteration} with $\beta=1$ for an
arbitrary starting vector $\x_{0}$, that is iteration~\eqref{eq:a-iteration} 
where $\d_{0},\d_{1},\dots$ are
i.i.d. vectors that take $\e^{(1)},\dots,\e^{(n)}$ with equal probability,
and $k(0),k(1),\dots$ are such that~\eqref{eq:delay} holds but
are independent of the random choices of $\d_{0},\d_{1},\dots$. Let
$\rho= \frac{1}{n}\InfNorm{\matA}=\max_{l}\left\{ \frac{1}{n}\sum_{r=1}^{n}\left|\matA_{lr}\right|\right\} $.
Provided that $2\rho\tau<1$, the following holds:
\begin{enumerate}
\item [(a)]For every $m\geq\frac{\log(1/2)}{\log(1-\lambda_{\max}/n)}\approx\frac{0.693n}{\lambda_{\max}}$
we have
\[
E_{m}\leq\left(1-\frac{\nu_{\tau}}{2\kappa}\right)E_{0}\,,
\]
where
\[
\nu_{\tau}=1-2\rho\tau
\]

\item [(b)]Let $T_{0}=\left\lceil \frac{\log(1/2)}{\log(1-\lambda_{\max}/n)}\right\rceil $
and $T=T_{0}+\tau$. For every $m\geq rT$ ($r=1,2,\dots$ ) we have
{\small{}
\[
E_{m}\leq\left(1-\frac{\nu_{\tau}}{2\kappa}\right)\left(1-\frac{\nu_{\tau}(1-\lambda_{\max}/n)^{\tau}}{2\kappa}+\chi\right)^{r-1}E_{0}
\]
}where
\[
\chi=\frac{\rho\tau^{2}\lambda_{\max}(1-\lambda_{\max}/n){}^{-2\tau}}{n}\,.
\]

\end{enumerate}
\end{thm}

\begin{proof}
In the proof, we use the following abbreviations:
\[
\delta_{\min}=\frac{\nu_{\tau}\lambda_{\min}}{n},\qquad\delta_{\max}=1-\frac{\lambda_{\max}}{n}\,.
\]

We begin with simple algebraic manipulations:

\begin{eqnarray}
\ANormS{\x_{j+1}-\x^{\star}} & = & \ANormS{\x_{j}+\gamma_{j}\d_{j}-\x^{\star}}\nonumber \\
 & = & \ANormS{\x_{j}-\x^{\star}}+\ANormS{\gamma_{j}\d_{j}}+2\dotprod{\x_{j}-\x^{\star}}{\gamma_{j}\d_{j}}{\matA}\nonumber \\
 & = & \ANormS{\x_{j}-\x^{\star}}+\gamma_{j}^{2}+2\gamma_{j}\dotprod{\x_{j}-\x^{\star}}{\d_{j}}{\matA}\nonumber \\
 & = & \ANormS{\x_{j}-\x^{\star}}+\dotprodsqr{\x_{k(j)}-\x^{\star}}{\d_{j}}{\matA}-2\dotprod{\x_{k(j)}-\x^{\star}}{\d_{j}}{\matA}\dotprod{\x_{j}-\x^{\star}}{\d_{j}}{\matA}\nonumber \\
 & = & \ANormS{\x_{j}-\x^{\star}}+\dotprodsqr{\x_{k(j)}-\x^{\star}}{\d_{j}}{\matA}\nonumber \\
 &  & \qquad-2\dotprod{\x_{k(j)}-\x^{\star}}{\d_{j}}{\matA}\left[\dotprod{\x_{j}-\x_{k(j)}}{\d_{j}}{\matA}+\dotprod{\x_{k(j)}-\x^{\star}}{\d_{j}}{\matA}\right]\nonumber \\
 & = & \ANormS{\x_{j}-\x^{\star}}-\dotprodsqr{\x_{k(j)}-\x^{\star}}{\d_{j}}{\matA}-2\dotprod{\x_{k(j)}-\x^{\star}}{\d_{j}}{\matA}\dotprod{\x_{j}-\x_{k(j)}}{\d_{j}}{\matA}\label{eq:baseline}
\end{eqnarray}
In the above we use the fact that $\matA$ has unit diagonal, so $\dotprod{\d_{i}}{\d_{i}}{\matA}=1$
for all $i$. We see that the error decreases by a positive ``progress
term'' ($\dotprodsqr{\x_{k(j)}-\x^{\star}}{\d_{j}}{\matA}$), and
it changes by an additional term ($2\dotprod{\x_{k(j)}-\x^{\star}}{\d_{j}}{\matA}\dotprod{\x_{k(j)}-\x_{j}}{\d_{j}}{\matA}$),
which might be positive or negative. When the iterations are synchronized
($k(j)=j$), there is no additional term, and the analysis reduces
to the analysis of synchronous randomized Gauss-Seidel.

We first bound the additional term:

\begin{eqnarray}
2\dotprod{\x_{k(j)}-\x^{\star}}{\d_{j}}{\matA}\dotprod{\x_{j}-\x_{k(j)}}{\d_{j}}{\matA} & = & 2\dotprod{\x_{k(j)}-\x^{\star}}{\d_{j}}{\matA}\dotprod{\sum_{t=k(j)}^{j-1}\gamma_{t}\d_{t}}{\d_{j}}{\matA}\nonumber \\
 & = & \sum_{t=k(j)}^{j-1}2\dotprod{\x_{k(j)}-\x^{\star}}{\d_{j}}{\matA}\dotprod{\x^{\star}-\x_{k(t)}}{\d_{t}}{\matA}\dotprod{\d_{t}}{\d_{j}}{\matA}\label{eq:cross}\\
 & \geq & -\sum_{t=k(j)}^{j-1}\left[\dotprodsqr{\x_{k(j)}-\x^{\star}}{\d_{j}}{\matA}\left|\dotprod{\d_{t}}{\d_{j}}{\matA}\right|+\dotprodsqr{\x_{k(t)}-\x^{\star}}{\d_{t}}{\matA}\left|\dotprod{\d_{t}}{\d_{j}}{\matA}\right|\right]\,.\nonumber 
\end{eqnarray}
Since $k(j)\leq t<j$:

\begin{eqnarray*}
\Expect{\dotprodsqr{\x_{k(j)}-\x^{\star}}{\d_{j}}{\matA}\left|\dotprod{\d_{t}}{\d_{j}}{\matA}\right|} & = & \Expect{\Expect{\dotprodsqr{\x_{k(j)}-\x^{\star}}{\d_{j}}{\matA}\left|\dotprod{\d_{t}}{\d_{j}}{\matA}\right|\middle|\,\d_{0},\dots,\d_{t-1}}}\\
 & = & \Expect{\frac{1}{n^{2}}\sum_{l=1}^{n}\sum_{r=1}^{n}\dotprodsqr{\x_{k(j)}-\x^{\star}}{\e^{(l)}}{\matA}\left|\dotprod{\e^{(l)}}{\e^{(r)}}{\matA}\right|}\\
 & = & \Expect{\frac{1}{n^{2}}\sum_{l=1}^{n}\sum_{r=1}^{n}\dotprodsqr{\x_{k(j)}-\x^{\star}}{\e^{(l)}}{\matA}\left|\matA_{lr}\right|}\\
 & \leq & \rho\Expect{\frac{1}{n}\sum_{l=1}^{n}\dotprodsqr{\x_{k(j)}-\x^{\star}}{\e^{(l)}}{\matA}}=\rho\Expect{\dotprodsqr{\x_{k(j)}-\x^{\star}}{\d_{j}}{\matA}}\,.
\end{eqnarray*}
Similarly, $\Expect{\dotprodsqr{\x_{k(t)}-\x^{\star}}{\d_{t}}{\matA}\left|\dotprod{\d_{t}}{\d_{j}}{\matA}\right|}\leq\rho\Expect{\dotprodsqr{\x_{k(t)}-\x^{\star}}{\d_{t}}{\matA}}$.
Taking expectation of~\eqref{eq:cross} and applying the last inequality
we find that
\begin{eqnarray*}
\Expect{2\dotprod{\x_{k(j)}-\x^{\star}}{\d_{j}}{\matA}\dotprod{\x_{j}-\x_{k(j)}}{\d_{j}}{\matA}} & \geq & -\rho\sum_{t=k(j)}^{j-1}\left[\Expect{\dotprodsqr{\x_{k(j)}-\x^{\star}}{\d_{j}}{\matA}}+\Expect{\dotprodsqr{\x_{k(t)}-\x^{\star}}{\d_{t}}{\matA}}\right]\\
 & = & -\rho|j-k(j)|\Expect{\dotprodsqr{\x_{k(j)}-\x^{\star}}{\d_{j}}{\matA}}-\rho\sum_{t=k(j)}^{j-1}\Expect{\dotprodsqr{\x_{k(t)}-\x^{\star}}{\d_{t}}{\matA}}\\
 & \geq & -\rho\tau\Expect{\dotprodsqr{\x_{k(j)}-\x^{\star}}{\d_{j}}{\matA}}-\rho\sum_{t=k(j)}^{j-1}\Expect{\dotprodsqr{\x_{k(t)}-\x^{\star}}{\d_{t}}{\matA}}\,.
\end{eqnarray*}
Taking expectation of \eqref{eq:baseline}, and plugging in the last
inequality, we find that
\begin{equation}
E_{j+1}\leq E_{j}-(1-\rho\tau)\Expect{\dotprodsqr{\x_{k(j)}-\x^{\star}}{\d_{j}}{\matA}}+\rho\sum_{t=k(j)}^{j-1}\Expect{\dotprodsqr{\x_{k(t)}-\x^{\star}}{\d_{t}}{\matA}}\,.\label{eq:err-recursion}
\end{equation}
Unrolling the recursion, we find that for every $m$:
\[
E_{m}\leq E_{0}-(1-\rho\tau)\sum_{i=0}^{m-1}\Expect{\dotprodsqr{\x_{k(i)}-\x^{\star}}{\d_{i}}{\matA}}+\rho\sum_{i=0}^{m-1}\sum_{t=k(i)}^{i-1}\Expect{\dotprodsqr{\x_{k(t)}-\x^{\star}}{\d_{t}}{\matA}}\,.
\]
In the last sum of the previous inequality ($\rho\sum_{i=0}^{m-1}\sum_{t=k(i)}^{i-1}\Expect{\dotprodsqr{\x_{k(t)}-\x^{\star}}{\d_{t}}{\matA}}$),
each term of the form $\Expect{\dotprodsqr{\x_{k(r)}-\x^{\star}}{\d_{r}}{\matA}}$
appears at most $\tau$ times, each time with a coefficient $\rho$.
So
\[
E_{m}\leq E_{0}-(1-2\rho\tau)\sum_{i=0}^{m-1}\Expect{\dotprodsqr{\x_{k(i)}-\x^{\star}}{\d_{i}}{\matA}}\,.
\]
We now apply the bound $\Expect{\dotprodsqr{\x_{k(i)}-\x^{\star}}{\d_{i}}{\matA}}\geq(\lambda_{\min}/n)E_{k(i)}$
(Lemma~\ref{lem:key}), to find that
\begin{equation}
E_{m}\leq E_{0}-\delta_{\min}\sum_{i=0}^{m-1}E_{k(i)}\,.\label{eq:ubound}
\end{equation}

\noindent \textbf{Proof of (a). }Lemma~\ref{lem:key} implies that
for any $b\geq a$ we have $E_{b}\geq\delta_{\max}^{b-a}E_{a}$. Indeed,

\begin{eqnarray*}
\ANormS{\x_{j+1}-\x^{\star}} & = & \ANormS{\x_{j}+\gamma_{j}\d_{j}-\x^{\star}}\\
 & = & \ANormS{\x_{j}-\x^{\star}}+\gamma_{j}^{2}+2\gamma_{j}\dotprod{\x_{j}-\x^{\star}}{\d_{j}}{\matA}\\
 & = & \ANormS{\x_{j}-\x^{\star}}+\dotprodsqr{\x_{k(j)}-\x^{\star}}{\d_{j}}{\matA}-2\dotprod{\x_{k(j)}-\x^{\star}}{\d_{j}}{\matA}\dotprod{\x_{j}-\x^{\star}}{\d_{j}}{\matA}\\
 & = & \ANormS{\x_{j}-\x^{\star}}+\dotprodsqr{\x_{k(j)}-\x_{j}+\x_{j}-\x^{\star}}{\d_{j}}{\matA}-2\dotprod{\x_{k(j)}-\x^{\star}}{\d_{j}}{\matA}\dotprod{\x_{j}-\x^{\star}}{\d_{j}}{\matA}\\
 & = & \ANormS{\x_{j}-\x^{\star}}+\dotprodsqr{\x_{k(j)}-\x^{\star}}{\d_{j}}{\matA}+\dotprodsqr{\x_{k(j)}-\x_{j}}{\d_{j}}{\matA}\\
 &  & \qquad+2\dotprod{\x_{j}-\x^{\star}}{\d_{j}}{\matA}\dotprod{\x_{k(j)}-\x_{j}}{\d_{j}}{\matA}\\
 &  & \qquad-2\dotprod{\x_{k(j)}-\x^{\star}}{\d_{j}}{\matA}\dotprod{\x_{j}-\x^{\star}}{\d_{j}}{\matA}\\
 & = & \ANormS{\x_{j}-\x^{\star}}+\dotprodsqr{\x_{j}-\x^{\star}}{\d_{j}}{\matA}+\dotprodsqr{\x_{k(j)}-\x_{j}}{\d_{j}}{\matA}\\
 &  & \qquad-2\dotprod{\x_{j}-\x^{\star}}{\d_{j}}{\matA}\dotprod{\x_{j}-\x^{\star}}{\d_{j}}{\matA}\\
 & = & \ANormS{\x_{j}-\x^{\star}}-\dotprodsqr{\x_{j}-\x^{\star}}{\d_{j}}{\matA}+\dotprodsqr{\x_{k(j)}-\x_{j}}{\d_{j}}{\matA}\\
 & \geq & \ANormS{\x_{j}-\x^{\star}}-\dotprodsqr{\x_{j}-\x^{\star}}{\d_{j}}{\matA}\,,
\end{eqnarray*}

\noindent so by taking expectation and applying Lemma~\ref{lem:key}
(notice that $\x_{k(j)}$ is independent of $\d_{j}$), we find that
$E_{j+1}\geq\delta_{\max}E_{j}$. In particular since $i\geq k(i)$,
\begin{equation}
E_{k(i)}\geq\delta_{\max}^{k(i)}E_{0}\geq\delta_{\max}^{i}E_{0}\,.\label{eq:lbound}
\end{equation}
Plugging~\eqref{eq:lbound} into \eqref{eq:ubound} we get the following
inequality, which leads immediately to assertion (a):
\[
E_{m}\leq\left(1-\delta_{\min}\sum_{i=0}^{m-1}\delta_{\max}^{i}\right)E_{0}=\left(1-\frac{\delta_{\min}(1-\delta_{\max}^{m})}{1-\delta_{\max}}\right)E_{0}=(1-\nu_{\tau}\kappa^{-1}(1-\delta_{\max}^{m}))E_{0}\,.
\]

\noindent \textbf{Proof of (b).} Let
\[
C_{i}=\left\{ rT+i-\tau\leq t\leq rT+i-1\,:\, t\geq rT\right\}
\]
and
\[
D_{i}=\left\{ rT+i-\tau\leq t\leq rT+i-1\,:\, t<rT\right\} \,.
\]

\noindent Unrolling the recursion in equation~\eqref{eq:err-recursion}
starting at $rT$, we find that for $r\geq1$ and $w\geq0$

\noindent
\begin{eqnarray}
E_{(r+1)T+w} & \leq & E_{rT}-(1-\rho\tau)\sum_{i=0}^{T-1+w}\Expect{\dotprodsqr{\x_{k(rT+i)}-\x^{\star}}{\d_{rT+i}}{\matA}}+\rho\sum_{i=0}^{T-1+w}\sum_{t=k(rT+i)}^{rT+i-1}\Expect{\dotprodsqr{\x_{k(t)}-\x^{\star}}{\d_{t}}{\matA}}\nonumber \\
 & \leq & E_{rT}-(1-\rho\tau)\sum_{i=0}^{T-1+w}\Expect{\dotprodsqr{\x_{k(rT+i)}-\x^{\star}}{\d_{rT+i}}{\matA}}+\rho\sum_{i=0}^{T-1+w}\sum_{t=rT+i-\tau}^{rT+i-1}\Expect{\dotprodsqr{\x_{k(t)}-\x^{\star}}{\d_{t}}{\matA}}\nonumber \\
 & \leq & E_{rT}-(1-\rho\tau)\sum_{i=0}^{T-1+w}\Expect{\dotprodsqr{\x_{k(rT+i)}-\x^{\star}}{\d_{rT+i}}{\matA}}+\rho\sum_{i=0}^{T-1+w}\sum_{t\in C_{i}}\Expect{\dotprodsqr{\x_{k(t)}-\x^{\star}}{\d_{t}}{\matA}}\nonumber \\
 &  & \qquad+\rho\sum_{i=0}^{T-1+w}\sum_{t\in D_{i}}\Expect{\dotprodsqr{\x_{k(t)}-\x^{\star}}{\d_{t}}{\matA}}\nonumber \\
 & \leq & E_{rT}-(1-2\rho\tau)\sum_{i=0}^{T-1+w}\Expect{\dotprodsqr{\x_{k(rT+i)}-\x^{\star}}{\d_{rT+i}}{\matA}}+\rho\sum_{i=0}^{\tau-1}\sum_{t\in D_{i}}\Expect{\dotprodsqr{\x_{k(t)}-\x^{\star}}{\d_{t}}{\matA}}\,.\nonumber \\
 & \leq & E_{rT}-(1-2\rho\tau)\sum_{i=\tau}^{T-1+w}\Expect{\dotprodsqr{\x_{k(rT+i)}-\x^{\star}}{\d_{rT+i}}{\matA}}+\rho\sum_{i=0}^{\tau-1}\sum_{t\in D_{i}}\Expect{\dotprodsqr{\x_{k(t)}-\x^{\star}}{\d_{t}}{\matA}}\,.\label{eq:bigequation}
\end{eqnarray}
The second-to-last inequality follows from the fact that each term
of the form $\Expect{\dotprodsqr{\x_{k(l)}-\x^{\star}}{\d_{l}}{\matA}}$
appears at most $\tau$ times in $\rho\sum_{i=0}^{T-1+w}\sum_{t\in C_{i}}\Expect{\dotprodsqr{\x_{k(t)}-\x^{\star}}{\d_{t}}{\matA}}$
. We also use the fact that for $i\geq\tau$ we trivially have $D_{i}=\emptyset$.

\noindent We first bound $E_{rT}-(1-2\rho\tau)\sum_{i=\tau}^{T-1+w}\Expect{\dotprodsqr{\x_{k(rT+i)}-\x^{\star}}{\d_{rT+i}}{\matA}}$.
Using Lemma~\ref{lem:key},
\[
E_{rT}-(1-2\rho\tau)\sum_{i=\tau}^{T-1+w}\Expect{\dotprodsqr{\x_{k(rT+i)}-\x^{\star}}{\d_{rT+i}}{\matA}}\leq E_{rT}-\delta_{\min}\sum_{i=\tau}^{T-1+w}E_{k(rT+i)}\,.
\]
Since $i\geq\tau$ we have $k(kT+i)\geq kT$ so $E_{k(rT+i)}\geq\delta_{\max}^{k(rT+i)-rT}E_{rT}\geq\delta_{\max}^{i}E_{rT}$.
Therefore
\[
E_{rT}-(1-2\rho\tau)\sum_{i=\tau}^{T-1+w}\Expect{\dotprodsqr{\x_{k(rT+i)}-\x^{\star}}{\d_{rT+i}}{\matA}}\leq(1-\delta_{\min}\delta_{\max}^{\tau}\sum_{i=0}^{T-1+w-\tau}\delta_{\max}^{i})E_{rT}\,.
\]
Noticing that $T-1+w-\tau=T_{0}-1+w$ and bounding the geometric sum
as in assertion (a), we find that $(1-\delta_{\min}\delta_{\max}^{\tau}\sum_{i=0}^{T-1+w-\tau}\delta_{\max}^{i})\leq(1-\frac{\delta_{\max}^{\tau}\nu_{\tau}}{2\kappa})$,
so
\begin{equation}
E_{rT}-(1-2\rho\tau)\sum_{i=\tau}^{T-1+w}\Expect{\dotprodsqr{\x_{k(rT+i)}-\x^{\star}}{\d_{rT+i}}{\matA}}\leq(1-\frac{\delta_{\max}^{\tau}\nu_{\tau}}{2\kappa})E_{rT}\,.\label{eq:Ekt_first}
\end{equation}
We now bound $\rho\sum_{i=0}^{\tau-1}\sum_{t\in D_{i}}\Expect{\dotprodsqr{\x_{k(t)}-\x^{\star}}{\d_{t}}{\matA}}$.
Recall that for every $b\geq a$ we have $E_{b}\geq\delta_{\max}^{b-a}E_{a},$
so, for $i=0,\dots,\tau-1$ and $t\in D_{i}$ we have
\[
E_{k(t)}\leq\delta_{\max}^{k(t)-rT}E_{rT}\leq\delta_{\max}^{-2\tau}E_{rT}\,.
\]
The last inequality follows from the fact that for $t\in D_{i}$,
$k(t)-rT\geq-2\tau$ and $\delta_{\max}<1$. We now bound
\[
\rho\sum_{i=0}^{\tau-1}\sum_{t\in D_{i}}\Expect{\dotprodsqr{\x_{k(t)}-\x^{\star}}{\d_{t}}{\matA}}\leq\rho\sum_{i=0}^{\tau-1}\sum_{t\in D_{i}}\frac{\lambda_{\max}\delta_{\max}^{-2\tau}}{n}E_{rT}\leq\frac{\rho\tau^{2}\lambda_{\max}\delta_{\max}^{-2\tau}}{n}E_{rT}=\chi E_{rT}\,.
\]
Combine the last inequality with~\eqref{eq:Ekt_first} and assertion
(a) to complete the proof of assertion (b).
\end{proof}
\textbf{Discussion:}
\begin{itemize}
\item Assertion (a) shows that after we perform enough asynchronous iterations,
we are guaranteed to reduce the expected error by a constant factor.
In order to drive the expected error down to an arbitrary fraction
of the input error, we can adopt the following scheme. We start with
asynchronous iterations. After $n$ iterations have been completed
we synchronize the threads and restart the iterations. The matrix
$\matA$ has unit diagonal, so $\lambda_{\max}\geq1$. Therefore,
by performing $k\geq n$ iterations, we are guaranteeing a $1-\nu_{\tau}/2\kappa$
factor reduction in the expected error. We then continue to iterate
and synchronize until the expected error is guaranteed to be small
enough. The number of outer iterations until convergence (reduction
of error by a predetermined factor) is $O(\kappa/\nu_{\tau})$. This
is also the number of synchronization points. When $\nu_{\tau}$ is
close to one, the number of synchronization points is asymptotically
the same as in Jacobi, but the convergence rate is that of  Gauss-Seidel.
Furthermore, we do not need to really divide the iterations between
processors (basically, every processor can do as many iterations as
it can, until synchronization) and it is not important to synchronize
exactly after $n$ iterations. So, from a practical perspective, a
time based scheme for synchronizing the processors should be sufficient,
and will not suffer from large wait times due to load imbalance.
\item Assertion (b) shows that even if we do not occasionally synchronize
the threads, we still get long-term linear convergence, but at a slower
rate. We say convergence is linear in the long term since we cannot
guarantee a diminishing bound in every iteration, but we can prove
a constant factor reduction over a large enough number of iterations.
\item The terms $\delta_{\max}^{\tau}$ and $\delta_{\max}^{-2\tau}$ that
appear in assertion (b) might seem problematic as they are exponential
in the number of processors (because $\tau=\Omega(P))$. However,
in our reference scenario this is not an issue because
$\lambda_{\max}=O(1)$ and $\tau\ll n$ (since we assume that $P\ll n$),
so $\delta_{\max}^{\tau}$ and $\delta_{\max}^{-2\tau}$ are actually
very close to $1$.
\item The number of iterations to guarantee a $1-\nu_{\tau}/2\kappa$ reduction
of expected error (as in assertion (a)) in synchronous randomized
Gauss-Seidel is approximately $\nu_{\tau}n/2\lambda_{\max}$. When
$\nu_{\tau}$ is close to one this is only slightly better than the
bound for AsyRGS, which is a small price to pay for the good speedups
expected for the asynchronous algorithm.
\item Consider our reference scenario in a weak-scaling regime (i.e., $P\approx cn$
for a very small $c$). In this case, $\nu_{\tau}$ is bounded by
a constant close to one because $\rho=O(1/n)$ in the reference scenario.
Therefore, with occasional synchronization of the threads,
the number of iterations increases by a small constant factor
due to asynchronism.
That is, the asynchronous phases do not violate the weak-scaling,
although the number of iterations can increase due to $\lambda_{\min}$
becoming smaller. As for the case where only asynchronous iterations
are used, we have $\chi\approx c^{2}\lambda_{\max}^{2}$. So, $\chi$
itself exhibits weak scaling. However, its value should be interpreted
with respect to $\kappa^{-1}$. If $\lambda_{\min}$ shrinks as $n$
grows, as is the case in many applications, then the relative size
of $\chi$ grows and we do not have weak scaling.
\item In general, if $\rho=O(1/n)$, we have $\nu_{\tau}=O(1)$ and the
discussion in the previous paragraph applies. Sparsity is not the
only scenario in which $\rho=O(1/n)$: for example, $\rho\leq2/n$
if $\matA$ is symmetric diagonally dominant, regardless of sparsity.
Other strong decay properties of off-diagonal entries might guarantee
$\rho=O(1/n)$ as well.
\end{itemize}

\section{\label{sec:step}Improving Scalability by Controlling Step-Size}

If we inspect~\eqref{eq:rgs-bound} we see that the best bound is 
attained for unit step-size. Griebel and Oswald introduced a step-size
since it is known that for certain applications over/under relaxations
converge faster~\cite{GO12}. In this section, we show that by 
controlling the step-size, it is possible to have a convergent method 
for any delay (as long as we set the step size small enough), unlike the
bound in Theorem~\ref{thm:main} which requires $2\rho\tau<1$. 
In addition, we show that by optimizing the step-size we can also improve the
scaling (dependence on $\tau$) in our bounds.

Our more general analysis starts with some simple algebraic
manipulations:

\begin{eqnarray}
\ANormS{\x_{j+1}-\x^{\star}} & = & \ANormS{\x_{j}+\beta\gamma_{j}\d_{j}-\x^{\star}}\nonumber \\
 & = & \ANormS{\x_{j}-\x^{\star}}+\ANormS{\beta\gamma_{j}\d_{j}}+2\dotprod{\x_{j}-\x^{\star}}{\beta\gamma_{j}\d_{j}}{\matA}\nonumber \\
 & = & \ANormS{\x_{j}-\x^{\star}}+\beta^{2}\gamma_{j}^{2}+2\beta\gamma_{j}\dotprod{\x_{j}-\x^{\star}}{\d_{j}}{\matA}\nonumber \\
 & = & \ANormS{\x_{j}-\x^{\star}}+\beta^{2}\dotprodsqr{\x_{k(j)}-\x^{\star}}{\d_{j}}{\matA}-2\beta\dotprod{\x_{k(j)}-\x^{\star}}{\d_{j}}{\matA}\dotprod{\x_{j}-\x^{\star}}{\d_{j}}{\matA}\nonumber \\
 & = & \ANormS{\x_{j}-\x^{\star}}+\beta^{2}\dotprodsqr{\x_{k(j)}-\x^{\star}}{\d_{j}}{\matA}\nonumber \\
 &  & \qquad-2\beta\dotprod{\x_{k(j)}-\x^{\star}}{\d_{j}}{\matA}\left[\dotprod{\x_{j}-\x_{k(j)}}{\d_{j}}{\matA}+\dotprod{\x_{k(j)}-\x^{\star}}{\d_{j}}{\matA}\right]\nonumber \\
 & = & \ANormS{\x_{j}-\x^{\star}}-\beta(2-\beta)\dotprodsqr{\x_{k(j)}-\x^{\star}}{\d_{j}}{\matA}-2\beta\dotprod{\x_{k(j)}-\x^{\star}}{\d_{j}}{\matA}\dotprod{\x_{j}-\x_{k(j)}}{\d_{j}}{\matA}\,.\label{eq:baseline-step}
\end{eqnarray}
As before, we continue with bounding the additional term.

\begin{eqnarray}
2\beta\dotprod{\x_{k(j)}-\x^{\star}}{\d_{j}}{\matA}\dotprod{\x_{j}-\x_{k(j)}}{\d_{j}}{\matA} & = & 2\beta\dotprod{\x_{k(j)}-\x^{\star}}{\d_{j}}{\matA}\dotprod{\sum_{t=k(j)}^{j-1}\beta\gamma_{t}\d_{t}}{\d_{j}}{\matA}\label{eq:add-consist}\\
 & = & \beta^{2}\sum_{t=k(j)}^{j-1}2\dotprod{\x_{k(j)}-\x^{\star}}{\d_{j}}{\matA}\dotprod{\x^{\star}-\x_{k(t)}}{\d_{t}}{\matA}\dotprod{\d_{t}}{\d_{j}}{\matA}\nonumber 
\end{eqnarray}
\vspace{-0.3in}
\[
\geq-\beta^{2}\sum_{t=k(j)}^{j-1}\left[\dotprodsqr{\x_{k(j)}-\x^{\star}}{\d_{j}}{\matA}\left|\dotprod{\d_{t}}{\d_{j}}{\matA}\right|\right.\left.+\dotprodsqr{\x_{k(t)}-\x^{\star}}{\d_{t}}{\matA}\left|\dotprod{\d_{t}}{\d_{j}}{\matA}\right|\right]\,.
\]
We see that the progress term is $O(\beta)$, but the additional term
is $O(\beta^{2}$). In synchronous randomized Gauss-Seidel the best
bound on the expected error is achieved with $\beta=1$, but for an
asynchronous computation the best bound is achieved with some $\beta<1$
(depending on $\tau$).

It is still the case that $E_{j+1}\geq\delta_{\max}E_{j}$. Indeed,
\begin{eqnarray*}
\x_{j+1} & = & \x_{j}+\beta\gamma_{j}\d_{j}\\
 & = & \x_{j}+\beta\dotprod{\x^{\star}-\x_{k(j)}}{\d_{j}}{\matA}\d_{j}\\
 & = & \x_{j}+\dotprod{\beta\x^{\star}-\beta\x_{k(j)}}{\d_{j}}{\matA}\d_{j}\\
 & = & \x_{j}+\dotprod{\x^{\star}-\y}{\d_{j}}{\matA}\d_{j}
\end{eqnarray*}
where $\y=(1-\beta)\x^{\star}+\x_{k(j)}$. Denote $\tilde{\gamma}_{j}=\dotprod{\x^{\star}-\y}{\d_{j}}{\matA}$.
Now,

\begin{eqnarray*}
\ANormS{\x_{j+1}-\x^{\star}} & = & \ANormS{\x_{j}+\tilde{\gamma}_{j}\d_{j}-\x^{\star}}\\
 & = & \ANormS{\x_{j}-\x^{\star}}+\ANormS{\tilde{\gamma}_{j}\d_{j}}+2\dotprod{\x_{j}-\x^{\star}}{\tilde{\gamma}_{j}\d_{j}}{\matA}\\
 & = & \ANormS{\x_{j}-\x^{\star}}+\tilde{\gamma}_{j}^{2}+2\tilde{\gamma}_{j}\dotprod{\x_{j}-\x^{\star}}{\d_{j}}{\matA}\\
 & = & \ANormS{\x_{j}-\x^{\star}}+\dotprodsqr{\y-\x^{\star}}{\d_{j}}{\matA}-2\dotprod{\y-\x^{\star}}{\d_{j}}{\matA}\dotprod{\x_{j}-\x^{\star}}{\d_{j}}{\matA}\\
 & = & \ANormS{\x_{j}-\x^{\star}}+\dotprodsqr{\y-\x_{j}+\x_{j}-\x^{\star}}{\d_{j}}{\matA}-2\dotprod{\y-\x^{\star}}{\d_{j}}{\matA}\dotprod{\x_{j}-\x^{\star}}{\d_{j}}{\matA}\\
 & = & \ANormS{\x_{j}-\x^{\star}}+\dotprodsqr{\x_{j}-\x^{\star}}{\d_{j}}{\matA}+\dotprodsqr{\y-\x_{j}}{\d_{j}}{\matA}\\
 &  & \qquad+2\dotprod{\x_{j}-\x^{\star}}{\d_{j}}{\matA}\dotprod{\y-\x_{j}}{\d_{j}}{\matA}\\
 &  & \qquad-2\dotprod{\y-\x^{\star}}{\d_{j}}{\matA}\dotprod{\x_{j}-\x^{\star}}{\d_{j}}{\matA}\\
 & = & \ANormS{\x_{j}-\x^{\star}}+\dotprodsqr{\x_{j}-\x^{\star}}{\d_{j}}{\matA}+\dotprodsqr{\y-\x_{j}}{\d_{j}}{\matA}\\
 &  & \qquad-2\dotprod{\x_{j}-\x^{\star}}{\d_{j}}{\matA}\dotprod{\x_{j}-\x^{\star}}{\d_{j}}{\matA}\\
 & = & \ANormS{\x_{j}-\x^{\star}}-\dotprodsqr{\x_{j}-\x^{\star}}{\d_{j}}{\matA}+\dotprodsqr{\y-\x_{j}}{\d_{j}}{\matA}\\
 & \geq & \ANormS{\x_{j}-\x^{\star}}-\dotprodsqr{\x_{j}-\x^{\star}}{\d_{j}}{\matA}\,.
\end{eqnarray*}
$E_{j+1}\geq\delta_{\max}E_{j}$ now follows by taking expectation
and applying Lemma~\ref{lem:key} (notice that $\y$ is independent
of $\d_{j}$),

Continuing along the lines of the proof of Theorem~\ref{thm:main}
(we omit the details), we have the following generalization of Theorem~\ref{thm:main}.

\begin{thm}
Consider iteration 
\eqref{eq:full-a-iteration} with $\beta\leq1$ for an
arbitrary starting vector $\x_{0}$, that is iteration~\eqref{eq:a-iteration} 
where $\d_{0},\d_{1},\dots$ are
i.i.d. vectors that take $\e^{(1)},\dots,\e^{(n)}$ with equal probability,
and $k(0),k(1),\dots$ are such that~\eqref{eq:delay} holds but
are independent of the random choices of $\d_{0},\d_{1},\dots$. Let
$\rho=\frac{1}{n}\InfNorm{\matA}=\max_{l}\left\{ \frac{1}{n}\sum_{r=1}^{n}\left|\matA_{lr}\right|\right\} $.
Provided that $2\beta-\beta^{2}-2\rho\tau\beta^{2}>0$, the following holds:
\begin{enumerate}
\item [(a)]For every $m\geq\frac{\log(1/2)}{\log(1-\lambda_{\max}/n)}\approx\frac{0.693n}{\lambda_{\max}}$
we have
\[
E_{m}\leq\left(1-\frac{\nu_{\tau}(\beta)}{2\kappa}\right)E_{0}\,,
\]
where
\[
\nu_{\tau}(\beta)=2\beta-\beta^{2}-2\rho\tau\beta^{2}
\]

\item [(b)]Let $T_{0}=\left\lceil \frac{\log(1/2)}{\log(1-\lambda_{\max}/n)}\right\rceil $
and $T=T_{0}+\tau$. For every $m\geq rT$ ($r=1,2,\dots$ ) we have
\[
E_{m}\leq\left(1-\frac{\nu_{\tau}(\beta)}{2\kappa}\right)\left(1-\frac{\nu_{\tau}(\beta)(1-\lambda_{\max}/n)^{\tau}}{2\kappa}+\chi(\beta)\right)^{r-1}E_{0}
\]
where
\[
\chi(\beta)=\frac{\rho\tau^{2}\beta^{2}\lambda_{\max}(1-\lambda_{\max}/n)^{-2\tau}}{n}\,.
\]

\end{enumerate}
\end{thm}

\textbf{Discussion:}
\begin{itemize}
\item We see that for a sufficiently small $\beta$ both bounds are useful,
but the computation of the optimal $\beta$ for assertion (b) (in
terms of the bound) requires some approximation of the condition number.
\item Alternatively, we can optimize only the value of $\nu_{\tau}(\beta)$.
The optimum of that term is achieved at $\tilde{\beta}=1/(1+2\rho\tau)$
and yields $\nu_{\tau}(\tilde{\beta})=1/(1+2\rho\tau)$. It is also
the case that $\chi(\tilde{\beta})<\chi(1)$, so both bounds are improved.
From a practical perspective, the challenge of setting the step size
to $\tilde{\beta}$ is that $\tau$ might not be known. However, under
normal circumstances (and in the reference scenario) we have $\tau=O(P)$,
which can provide a general guideline for setting the step-size.
\end{itemize}

\section{\label{sec:without}Convergence Bound with Inconsistent Reads}

We now analyze the iteration under the inconsistent read model, i.e. 
iteration~\eqref{eq:full-inc-iteration}. 
\begin{thm}
\label{thm:also}Consider iteration \eqref{eq:full-inc-iteration}
for some $0\leq\beta<1$ and an arbitrary starting vector $\x_{0}$,
where $\d_{0},\d_{1},\dots$ are i.i.d. vectors that take $\e^{(1)},\dots,\e^{(n)}$
with equal probability, and $K(0),K(1),\dots$ are such that equation~\eqref{eq:delay-K}
holds but are independent of the random choices of $\d_{0},\d_{1},\dots$.
Let $\rho_{2}=\max_{l}\left\{ \frac{1}{n}\sum_{r=1}^{n}\matA_{lr}^{2}\right\} $.
Provided that $\beta(1-\beta-\rho_{2}\tau^{2}\beta/2)>0$, the following
holds:
\begin{enumerate}
\item [(a)]For every $m\geq\frac{\log(1/2)}{\log(1-\lambda_{\max}/n)}\approx\frac{0.693n}{\lambda_{\max}}$
we have
\[
E_{m}\leq\left(1-\frac{\omega_{\tau}(\beta)}{2\kappa}\right)E_{0}
\]
where
\[
\omega_{\tau}(\beta)=2\beta(1-\beta-\rho_{2}\tau^{2}\beta/2)~.
\]

\item [(b)]Let $T_{0}=\left\lceil \frac{\log(1/2)}{\log(1-\lambda_{\max}/n)}\right\rceil $
and $T=T_{0}+\tau$. For every $m\geq rT$ ($r=1,2,\dots$ ) we have
\[
E_{m}\leq\left(1-\frac{\omega_{\tau}(\beta)}{2\kappa}\right)\left(1-\frac{\omega_{\tau}(\beta)(1-\lambda_{\max}/n)^{\tau}}{2\kappa}+\psi(\beta)\right)^{r-1}E_{0}
\]
where
\[
\psi(\beta)=\frac{\rho_{2}\tau^{3}\beta^{2}\lambda_{\max}(1-\lambda_{\max}/n){}^{-2\tau}}{n}\,.
\]

\end{enumerate}
\end{thm}
Most of the proof is analogous to the proof of Theorem~\ref{thm:main},
so we give only a sketch that focuses on the unique parts.

\begin{proof}
\emph{(Sketch) }As before:
\[
\ANormS{\x_{j+1}-\x^{\star}}=\ANormS{\x_{j}-\x^{\star}}-\beta(2-\beta)\dotprodsqr{\x_{K(j)}-\x^{\star}}{\d_{j}}{\matA}-2\beta\dotprod{\x_{K(j)}-\x^{\star}}{\d_{j}}{\matA}\dotprod{\x_{j}-\x_{K(j)}}{\d_{j}}{\matA}\,.
\]
We now bound the additional term:

\begin{eqnarray}
2\beta\dotprod{\x_{K(j)}-\x^{\star}}{\d_{j}}{\matA}\dotprod{\x_{j}-\x_{K(j)}}{\d_{j}}{\matA} & = & 2\beta\dotprod{\x_{K(j)}-\x^{\star}}{\d_{j}}{\matA}\dotprod{\sum_{t\in K^{-}(j)}\beta\gamma_{t}\d_{t}}{\d_{j}}{\matA}\nonumber \\
 & = & 2\beta^{2}\dotprod{\x_{K(j)}-\x^{\star}}{\d_{j}}{\matA}\dotprod{\sum_{t\in K^{-}(j)}\gamma_{t}\d_{t}}{\d_{j}}{\matA}\nonumber \\
 & \geq & -\beta^{2}\left[\dotprodsqr{\x_{K(j)}-\x^{\star}}{\d_{j}}{\matA}+\dotprodsqr{\sum_{t\in K^{-}(j)}\gamma_{t}\d_{t}}{\d_{j}}A\right]\label{eq:add-inconsist}\\
 & \geq & -\beta^{2}\left[\dotprodsqr{\x_{K(j)}-\x^{\star}}{\d_{j}}{\matA}+\right.\nonumber \\
 &  & \qquad\left.|K^{-}(j)|\sum_{t\in K^{-}(j)}\dotprodsqr{\x_{K(t)}-\x^{\star}}{\d_{t}}{\matA}\dotprodsqr{\d_{t}}{\d_{j}}{\matA}\right]\nonumber \\
 & \geq & -\beta^{2}\left[\dotprodsqr{\x_{K(j)}-\x^{\star}}{\d_{j}}{\matA}+\tau\sum_{t\in K^{-}(j)}\dotprodsqr{\x_{K(t)}-\x^{\star}}{\d_{t}}{\matA}\dotprodsqr{\d_{t}}{\d_{j}}{\matA}\right]\nonumber
\end{eqnarray}
where $K^{-}(j)=\{0,\dots,j-1\}-K(j)$. Since $\x_{K(t)}$ does not
depend on $\d_{t}$ or $\d_{j}$, we can bound as before,
\[
\Expect{\dotprodsqr{\x_{K(t)}-\x^{\star}}{\d_{t}}{\matA}\dotprodsqr{\d_{t}}{\d_{j}}{\matA}}\leq\rho_{2}\Expect{\dotprodsqr{\x_{K(t)}-\x^{\star}}{\d_{t}}{\matA}}\,.
\]
Therefore,
\[
E_{j+1}\leq E_{j}-2\beta(1-\beta)\Expect{\dotprodsqr{\x_{K(j)}-\x^{\star}}{\d_{j}}{\matA}}+\rho_{2}\tau\beta^{2}\sum_{t\in K^{-}(j)}\Expect{\dotprodsqr{\x_{K(t)}-\x^{\star}}{\d_{t}}{\matA}}\,.
\]
After we unroll the recursion, we find that
\[
E_{k}\leq E_{0}-2\beta(1-\beta-\rho_{2}\tau^{2}\beta/2)\sum_{i=0}^{k-1}\Expect{\dotprodsqr{\x_{K(i)}-\x^{\star}}{\d_{i}}{\matA}}\,.
\]
We can now continue to bound as in the proof of Theorem~\ref{thm:main}.
The crucial observation is that $\x_{K(i)}$ is the result of $|K(i)|$
random single coordinate steps. So
\[
\Expect{\Vert\x_{K(i)}-\x^{\star}\Vert_{\matA}^{2}}\ge\delta_{\max}^{|K(i)|}E_{0}\geq\delta_{\max}^{i}E_{0}\,.
\]

\end{proof}
\textbf{Discussion:}
\begin{itemize}
\item Unlike the bounds for the consistent read model, the theorem guarantees
convergence only for values of $\beta$ strictly smaller than $1$.

\item The bound has a worse dependence on $\tau$, so scalability is worse. The bound
for consistent read has also a better dependence on the step-size ($\beta$) in the sense
that large step-sizes are allowed.  In contrast, due to the unit diagonal assumption,
we always have $\rho_2 \leq \rho$, so the bound enjoys a more favorable dependence in that 
respect. However, we note that $\rho_2 \geq 1/n$, so if $\rho = O(1/n)$ then the ratio between 
$\rho$ and $\rho_2$ is bounded by a constant, so the better dependence on $\tau$ makes the 
bound for consistent read more favorable (as we expect).

\item The reason why we develop equation~\eqref{eq:add-inconsist} instead
of simply adapting equation~\eqref{eq:add-consist} for the inconsistent
read iteration is that the latter equation leads to expressions of
the form $\dotprodsqr{\x_{K(j)}-\x^{\star}}{\d_{j}}{\matA}\left|\dotprod{\d_{t}}{\d_{j}}{\matA}\right|$
for $t\in\{0,\dots,j-1\}-K(j)$. Such an expression is hard to analyze
since $\x_{K(j)}$ can depend of $\d_{t}$. An example is $K(j)=\{0,\dots,j-3,j-1\}$
and $t=j-2$ (for some $j\geq3$).
\end{itemize}

\section{\label{sec:LS}Unsymmetric Systems and Overdetermined Least-Squares}

In this section, we consider the more general problem of finding the
solution to $\min_{\x}\|\mat A\x-\b\|_{2}$ where $\mat A\in\mathbb{R}^{r\times n}$
has at least as many rows as columns and is full rank. Note that this
problem includes the solution of $\mat A\x=\b$ for a general (possibly
unsymmetric) non singular $\mat A$. For simplicity, we will assume
that the columns of $\mat A$ have unit Euclidean norm.

Lewis and Leventhal~\cite{LL10} analyzed this case as well, and
suggest the following iteration

\begin{equation}
\begin{alignedat}{1}\rb_{j} & =\,\b-\mat A\x_{j}\\
\gamma_{j} & =\,\d_{j}^{\T}\mat A^{\T}\rb_{j}\\
\x_{j+1} & =\,\x_{j}+\gamma_{j}\d_{j}
\end{alignedat}
\label{eq:sck}
\end{equation}
where $\d_{0},\d_{1},\dots$ are i.i.d. random vectors, taking $\e^{(1)},\dots,\e^{(n)}$
with equal probability. One can show that this is a stochastic coordinate
descent method on $f(\x)=\|\mat A\x-\b\|_{2}$. Lewis and Leventhal
prove that
\[
\Expect{\TNormS{\mat A\x_{m}-\b}}\leq\left(1-\frac{\lambda_{\min}(\matA^{\T}\mat A)}{n}\right)^{m}\TNormS{\matA\x_{0}-\b}\,,
\]
where in the above $\lambda_{\min}(\matA^{\T}\mat A)$ is the minimum
eigenvalue of $\mat A^{\T}\matA$.

It is rather straightforward to devise an asynchronous version of
the algorithm, following the same strategy we used for AsyRGS. We
only note that traditional presentations of~\eqref{eq:sck} favor
keeping a residual vector $\rb$ in memory, and updating after each
update on $\x$. However, updates to $\rb$ cannot be atomic, so in
an asynchronous version of the iteration, the necessary entries of
the residual (i.e., the non-zero indices of $\matA\d_{j}$) have to be 
computed in each iteration. Introducing
a step-size $\beta$ as well, leads to the following iteration that
describes the asynchronous algorithm (inconsistent read):
\begin{equation}
\begin{alignedat}{1}\gamma_{j} & =\,\d_{j}^{\T}\mat A^{\T}(\b-\mat A\x_{K(j)})\\
\x_{j+1} & =\,\x_{j}+\beta\gamma_{j}\d_{j}
\end{alignedat}
\label{eq:asck}
\end{equation}

We remark that each iteration of the asynchronous algorithm that implements~\eqref{eq:asck}
is more expensive then the best implementation of~\eqref{eq:sck}. The main cost per step 
of the asynchronous algorithm is in computing $\gamma_j$. Suppose that the row vector 
$\d_{j}^{\T}\mat A^{\T}$ (which is simply a transposed column of $\matA$) has $l_j$ non-zeros, and that the 
rows corresponding the non-zero indices in $\d_{j}^{\T}\mat A^{\T}$ have 
$r_{j,1},\dots,r_{j,l_j}$ non-zeros. The cost of iteration $j$ of~\eqref{eq:asck} is then $O(\sum^{l_j}_{i=1}r_{j,i})$.
As explained, for~\eqref{eq:sck}, we can have cheaper steps by keeping both $\x$ and $\rb$ in memory. 
Each iteration then updates both $\x$ and $\rb$ leading to a cost per iteration of 
$O(l_j)$. Since all rows have at least one non-zero, this is obviously superior. It is important
to note that while for some matrices the additional cost might be disastrous, it is not necessarily 
the case. For example, if the number of non-zeros in each column and row is bounded between $C_1$ and $C_2$,
each iteration is at most $O(C^2_2/C_1)$ more expensive. If both $C_2$ and $C_2/C_1$  are small, this ratio
might be small enough to be overcome by a sufficient amount of processors (in particular, if $C_1=O(1)$ and 
$C_2=O(1)$ then the asymptotic cost per iteration is unaffected).

We also remark that for~\eqref{eq:asck} to be a valid description of the behavior of the algorithm,
care has to be taken that in each iteration each entry of $\x$ that is read, is read only {\em once}.

Notice that~\eqref{eq:asck} is identical to the iteration of AsyRGS
on $\mat A^{\T}\matA\x=\matA^{\T}\b$. This leads immediately to the
following theorem.
\begin{thm}
\label{thm:also-1}Consider iteration \eqref{eq:asck} for some $0\leq\beta<1$
and an arbitrary starting vector $\x_{0}$, where $\d_{0},\d_{1},\dots$are
i.i.d. vectors that take $\e^{(1)},\dots,\e^{(n)}$ with equal probability,
and $K(0),K(1),\dots$ are such that equation~\eqref{eq:delay-K}
holds but are independent of the random choices of $\d_{0},\d_{1},\dots$.
Let $\mat X=\matA^{\T}\matA$ and let $\rho_{2}=\max_{l}\left\{ \frac{1}{n}\sum_{r=1}^{n}\mat X_{lr}^{2}\right\} $.
Let $\x^\star = \arg \min_{\x}\|\mat A\x-\b\|_{2}$.
Denote by $\kappa$ the condition number of $\matA$ (ratio between
largest and smallest singular values of $\mat A$) and by $\sigma_{\max}$
the maximum singular value of $\mat A$. Provided that $\beta(1-\beta-\rho_{2}\tau^{2}\beta/2)>0$,
the following holds:
\begin{enumerate}
\item [(a)]For every $m\geq\frac{\log(1/2)}{\log(1-\sigma_{\max}^{2}/n)}\approx\frac{0.693n}{\sigma_{\max}^{2}}$
we have
\[
\Expect{\XNormS{\x_{m}-\x^\star}}\leq\left(1-\frac{\omega_{\tau}(\beta)}{2\kappa^{2}}\right)\XNormS{\x_{0}-\x^\star}
\]
where
\[
\omega_{\tau}(\beta)=2\beta(1-\beta-\rho_{2}\tau^{2}\beta/2)
\]

\item [(b)]Let $T_{0}=\left\lceil \frac{\log(1/2)}{\log(1-\sigma_{\max}^{2}/n)}\right\rceil $
and $T=T_{0}+\tau$. For every $m\geq rT$ ($r=1,2,\dots$ ) we have
\[
\Expect{\XNormS{\x_{m}-\x^\star}}\leq\left(1-\frac{\omega_{\tau}(\beta)}{2\kappa^{2}}\right)\left(1-\frac{\omega_{\tau}(\beta)(1-\sigma_{\max}^{2}/n)^{\tau}}{2\kappa^{2}}+\chi(\beta)\right)^{r-1}\XNormS{\x_{0}-\x^\star}
\]
where
\[
\chi(\beta)=\frac{\rho_{2}\tau^{3}\beta^{2}\sigma_{\max}^{2}(1-\sigma_{\max}^{2}/n){}^{-2\tau}}{n}\,.
\]

\end{enumerate}

Note that when $\matA \x^\star = \b$ we have $\XNormS{\x_{m}-\x^\star} = \TNormS{\matA \x_m - \b}$ and $\XNormS{\x_{0}-\x^\star} = \TNormS{\matA \x_0 - \b}$.
\end{thm}
\section{\label{sec:experiments}Experiments}

The main goals of this section are threefold. First, we show that
the proposed algorithm can be advantageous for certain types of linear
systems, such as those arising in the analysis of big data, even in
the absence of massive parallelism. Secondly, we explore the behavior
of the algorithm in terms of scalability and the penalty paid for
asynchronicity for these type of matrices. Finally, we demonstrate
that our algorithm can serve as an effective preconditioner for a
flexible Krylov method.

It is not the goal of this section to show that the suggested algorithm
converges faster than standard algorithms like CG for all, or many,
matrices. The synchronous method on which our algorithm is based requires
$O(\kappa\cdot n)$ iterations for convergence, which are equivalent
to about $O(\kappa)$ CG iterations. In comparison, CG converges at
a much better rate of $O(\sqrt{\kappa})$. Therefore, for a general
purpose solver, our algorithm might be advantageous only as a preconditioner
in a Krylov subspace method that can handle a preconditioner that
changes from one step to another. Such methods are known as\emph{
flexible} Krylov subspace methods~\cite{Saad93,Notay00,SS03}. While
we explore this combination, it is not the purpose of this paper to
present a general purpose production-grade linear solver. For such
a solver, sophisticated rules for setting the preconditioner parameters,
and various heuristics will be needed to avoid a completely random
access pattern which is likely to cause poor performance due to extensive
cache misses, just to name a few of the issues that need to be resolved
for a production-grade linear solver. Exploration of these issues
is slated for continued future research, and is outside the scope
of this paper.

We experiment with a linear system arising from performing linear
regression to analyze social media data. The system arises from a
real-life data analysis task performed on real data. The matrix $\matA$
is $120{,}147\times120{,}147$ (after removing rows and columns that
were identically zero) and it has $172{,}892{,}749$ non-zeros. 
The matrix is generated by computing the Gram matrix of a data matrix
in which each row corresponds to a text document and the values
are the term frequencies inside the document. 
The maximum number of non-zeros in a row is $117{,}182$, so some rows
are almost full. However, the row sizes are highly skewed as the average
number of non-zeros in a row is $1439$, and the minimum is $1$.
The right-hand side $\mat B$ has $51$ columns, each corresponding to a label
prediction. For both CG and our
algorithm, the $120{,}147\times51$ right-hand side and solution matrices
are stored in a row-major fashion to improve locality. All $51$ systems
are solved together. We initialize the solution to an all zero matrix.
 The matrix does not have unit-diagonal, so for
Randomized Gauss-Seidel and AysRGS we use iteration~\eqref{eq:basic-iteration-nonunit}.
The coefficient matrix has very little to no
structure. This implies that reordering $\matA$ in order to reduce
cache misses in the matrix-vector multiplications has very little
effect. Luckily, the downstream application requires very low accuracy.
In particular, running Randomized Gauss-Seidel beyond $10$ sweeps
has negligible improvement in the downstream use (as measured in the
application specific metric), even though the residual continues to
drop.

For this matrix $\rho \approx 231 / n$ and $\rho_2 \approx 8.9 / n$,
so the parameters $\nu_\tau$ and $\omega_\tau$ that govern the slowdown
with occasional synchronization are not too bad (e.g., $\nu_{200}(1.0) = 0.618$
and $\omega_{200}(0.25) = 0.1906$). However matrix is highly ill-conditioned 
(as verified using an iterative condition-number estimator~\cite{ADT13})
so the bounds without occasional synchronization do not apply unless we set
the step size extremely small. In that respect, we note that the theoretical
bounds for the synchronous algorithm are already far from being descriptive
of the behavior of that algorithm, at least for a small amount of sweeps. 

We remark that we chose to experiment with this matrix even though it does not fit the reference scenario, which was the main focus of our theoretical analysis.  We wanted to experiment with an application for which our algorithm makes sense as a standalone algorithm, instead of the case where it is applied as a building block in a larger solver, which is the case for most matrices from scientific computing applications. In addition, we note that the high imbalance in the row size only makes this matrix more challenging for an asynchronous solver (the maximum delay $\tau$ might be large). The analysis presented in this paper is not able to cover effectively such high variance in row sizes, but we conjecture that the results can be strengthened to cover it, and propose a strategy in the conclusions. Finally, by choosing this test case, we are able to show that in practice, the algorithm can be suitable even for situations outside the scope of our analysis.

\begin{figure}[t]
\noindent \begin{centering}
\begin{tabular}{c}
\includegraphics[width=0.45\textwidth]{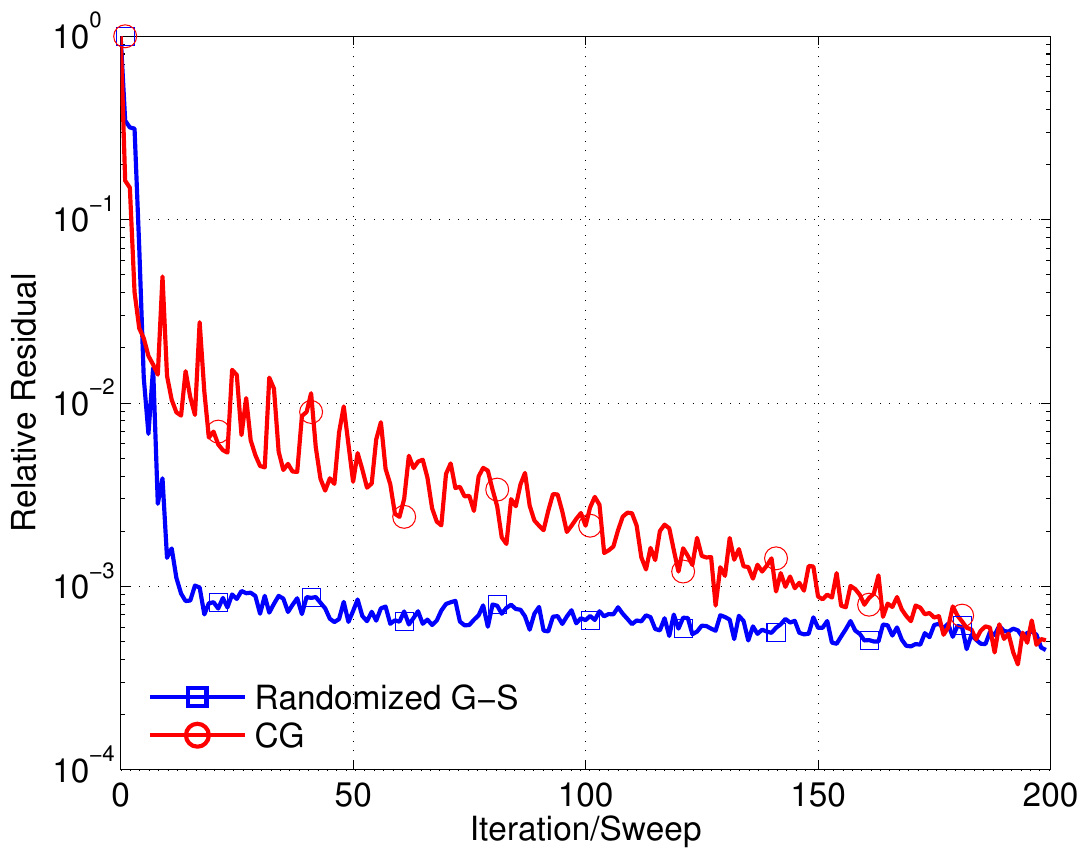}\tabularnewline
\end{tabular}
\par\end{centering}

\protect\caption{\label{fig:res}Residual of Randomized Gauss-Seidel and CG on the
test matrix. }
\end{figure}

Figure~\ref{fig:res} plots the residual ($\FNorm{\matA\mat X-\mat B}/\FNorm{\mat B}$)
of Randomized Gauss-Seidel and CG as the iterations progress. We see
that Randomized Gauss-Seidel initially progresses faster than CG.
This suggests that Randomized Gauss-Seidel,
and its asynchronous variants, might be well suited as a preconditioner
in a flexible Krylov method. We remark that the behavior of CG can
be improved with preconditioning.

\begin{figure}[t]
\noindent \begin{centering}
\begin{tabular}{ccc}
\includegraphics[width=0.3\textwidth]{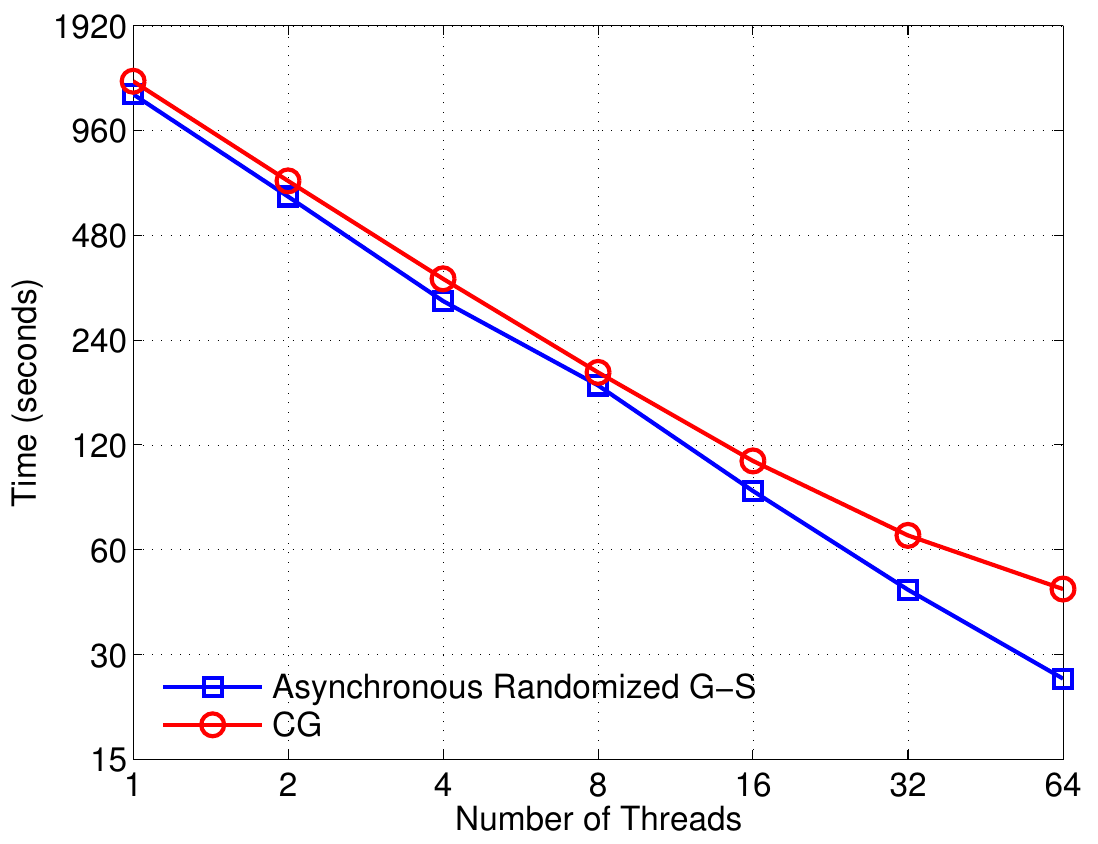} &  \includegraphics[width=0.3\textwidth]{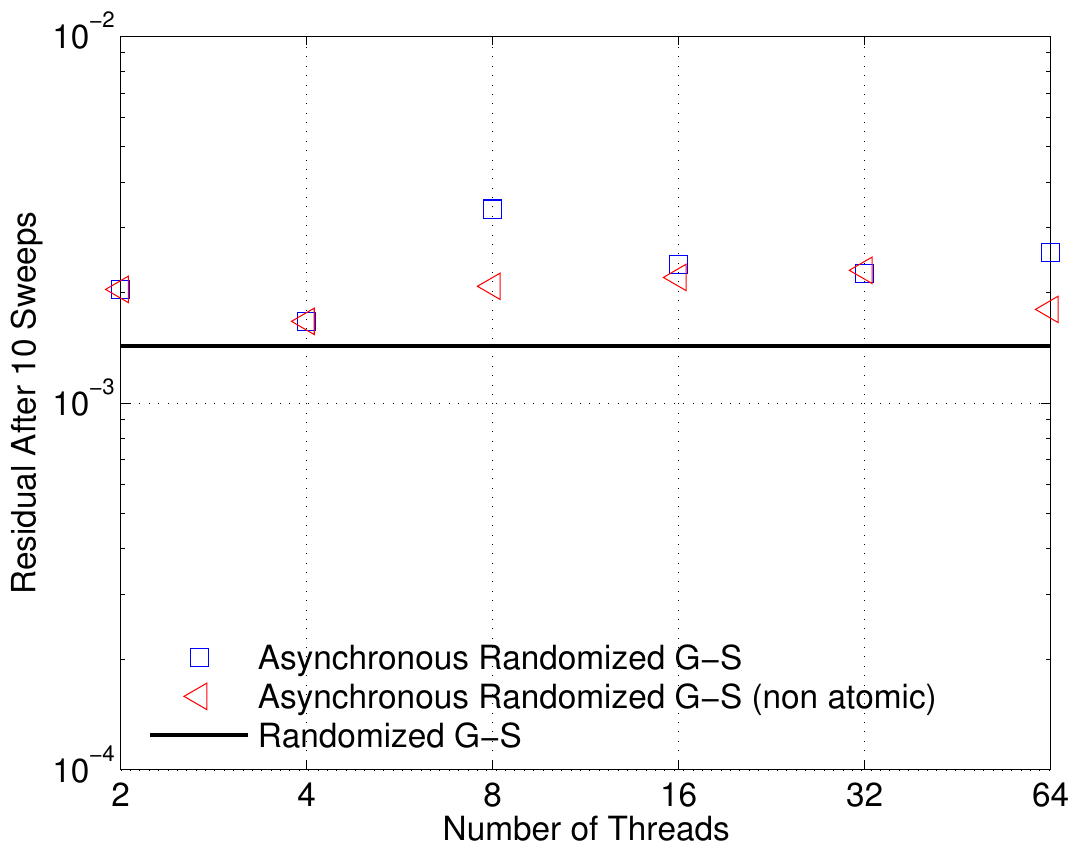} &  \includegraphics[width=0.3\textwidth]{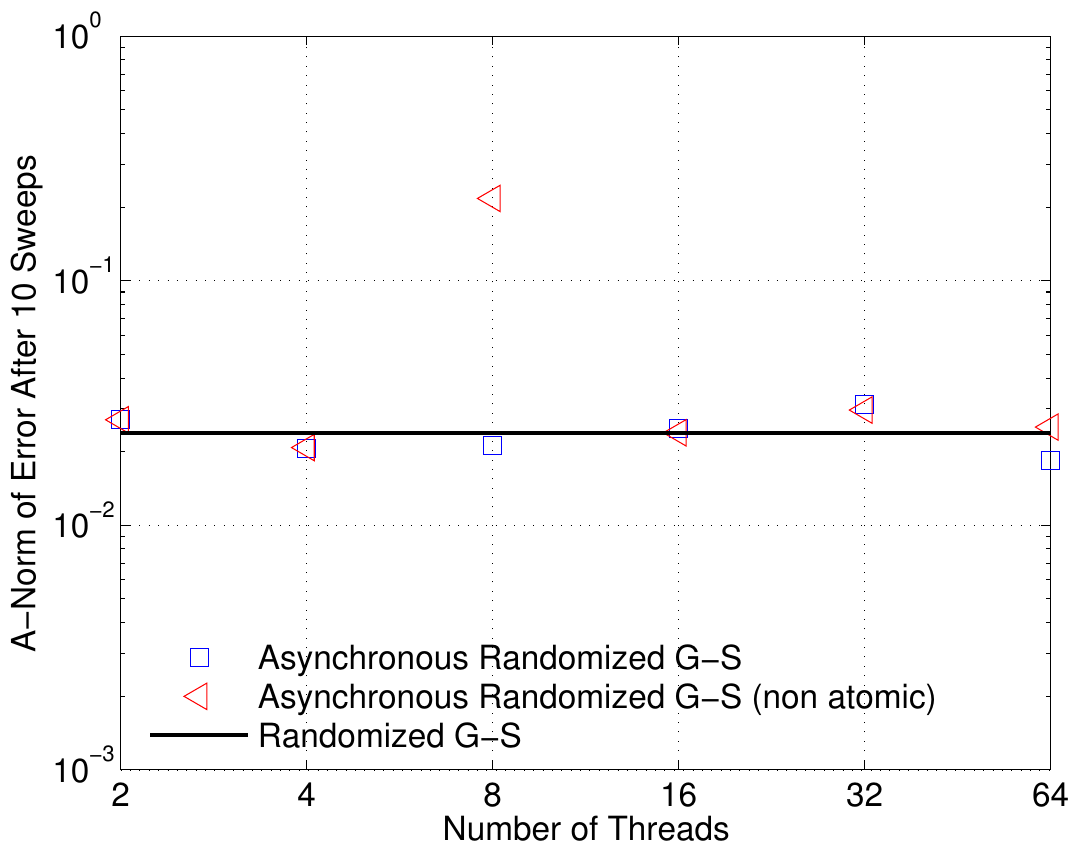} \tabularnewline
\end{tabular}
\par\end{centering}

\protect\caption{\label{fig:args}Performance of AsyRGS (inconsistent read) on the
test matrix. Left: Running time of both AsyRGS and CG on the test
matrix. Center: Relative residual after 10 sweeps of Randomized Gauss-Seidel
and AsyRGS. Right: Relative $\matA$-norm of the error.}
\end{figure}

We tested parallel performance on a single BlueGene/Q node. The compute
node has 16 compute cores (and an additional one for services) running
at 1.6 GHz, each capable of 4-way multithreading. We experimented
with the inconsistent read variant only. In Figure~\ref{fig:args}
(left), we plot the running time of $10$ iterations (sweeps) of AsyRGS
and CG. We use a SIMD variant of CG where the indices are assigned
to threads in a round-robin manner. The use of round-robin is due
to the fact that the coefficient matrix has very little to no structure,
so other distribution methods give very little benefit while incurring
a large overhead. We see that AsyRGS shows almost linear scalability,
and attains a speedup of almost $48$ on $64$ threads. CG initially
shows good speedups as well, but strays from linear speedup as the
thread count grows. In the serial run, Randomized Gauss-Seidel was
about $10\%$ faster ($1220$ seconds versus $1330$ seconds for CG).
With $64$ threads the gap is substantial: $25.7$ seconds versus
$46.5$ seconds for CG. The speedup of CG on $64$ threads is less
than $29$.

Next, we explore whether there is a price, in terms of the final residual,
for using an asynchronous version. To that end, we made sure that
the set of directions $\d_{0},\d_{1},\dots$ is fixed using the library
Random123~\cite{SMDS11} which allows random access to the pseudo-random
numbers, as opposed to the conventional streamed approach. Here we
also try a variant of AsyRGS which does not perform atomic writes,
in order to test experimentally whether atomic writes are necessary
(from a theoretical point-of-view, so far we have not been able to
analyze the convergence rate without atomic writes). In Figure~\ref{fig:args}
(center), we plot the residual after $10$ sweeps of a single run on
each thread count. We see that the residual of the asynchronous algorithm
is slightly worse than that of the synchronous method, although it
is of the same order of magnitude. There does not seem to be a consistent
advantage to using atomic writes. There is also variation in the residual
due to different scheduling of the threads, so we conducted $5$ additional
trials with $64$ threads. The minimum residual of AsyRGS was $1.44\times10^{-3}$
and the maximum was $2.88\times10^{-3}$. With the non-atomic variant,
the minimum was $1.39\times10^{-3}$ and the maximum was $2.96\times10^{-3}$.
There is no noticeable difference between the running time of the
two variants.

In Figure~\ref{fig:args} (right) we examine the relative $\matA$-norm of the error
($\ANorm{\x - \x^\star} / \ANorm{\x^\star}$) after $10$ sweeps for different thread counts. 
In this experiment, we  use a single right hand side that is generated as follows: we took one
of the right hand sides of the original problem, solved to very low residual
(using Flexible-CG, as explained in the following paragraphs) to form $\x^\star$, and 
then used $\b = \matA \x^\star$. We see that the $\matA$-norm of the error for AsyRGS
is very close to the $\matA$-norm of the error of the synchronous method, and
is sometimes better (we caution that we did a single experiment, and there is
variation in the error produced by the asynchronous method). Both the synchronous 
and the asynchronous method produce errors that are well below the theoretical 
bounds for the synchronous method.

\begin{table}
\protect\caption{\label{tab:args-fcg}Performance of Flexible-CG with AsyRGS (inconsistent
read) serving as a preconditoner. We explore the performance when
varying the number of inner (preconditioner) sweeps. $64$ threads
are used. As runs are not deterministic, we report the median of five
different runs. }

\centering{}%
\begin{tabular}{|c|c|c|c|c|}
\hline
Inner sweeps & Outer iterations & Outer $\times$ (Inner + 1) & Time & Mat-ops / sec\tabularnewline
\hline
\hline
30 & 38 & 1178 & 234 sec & 5.03\tabularnewline
\hline
20 & 48 & 1008 & 203 sec & 4.97\tabularnewline
\hline
10 & 69 & 759 & 159 sec & 4.77\tabularnewline
\hline
5 & 100 & 600 & 132 sec & 4.55\tabularnewline
\hline
3 & 151 & 604 & 134 sec & 4.51\tabularnewline
\hline
\textbf{2} & \textbf{184} & \textbf{552} & \textbf{125 sec} & \textbf{4.42}\tabularnewline
\hline
1 & 356 & 712 & 162 sec & 4.39\tabularnewline
\hline
\end{tabular}
\end{table}

Our final set of experiments explore the use of AsyRGS as a preconditioner
for a flexible Krylov method. We implemented Notay's Flexible-CG algorithm~\cite{Notay00}.
In our implementation we do not use truncation or restarts, although
we acknowledge that a general purpose production-grade solver might
require these. We compute the norm of the residual after every iteration
of Flexible-CG, and declare convergence once the relative residual
has dropped below some predefined threshold. In our experiments we
use $10^{-8}$. Iteration counts and running time are reported based
on convergence. We use AsyRGS as the preconditioner, with the number
of sweeps set as a parameter. As AsyRGS introduces non-determinism,
we repeat every experiment five times, and report the median (we note
that the random choices are fixed in these five runs, and non-determinism
is only due to asynchronism). We use only a single right-hand side
vector in this set of experiments.

\begin{figure}[t]
\noindent \begin{centering}
\begin{tabular}{ccc}
\includegraphics[width=0.45\textwidth]{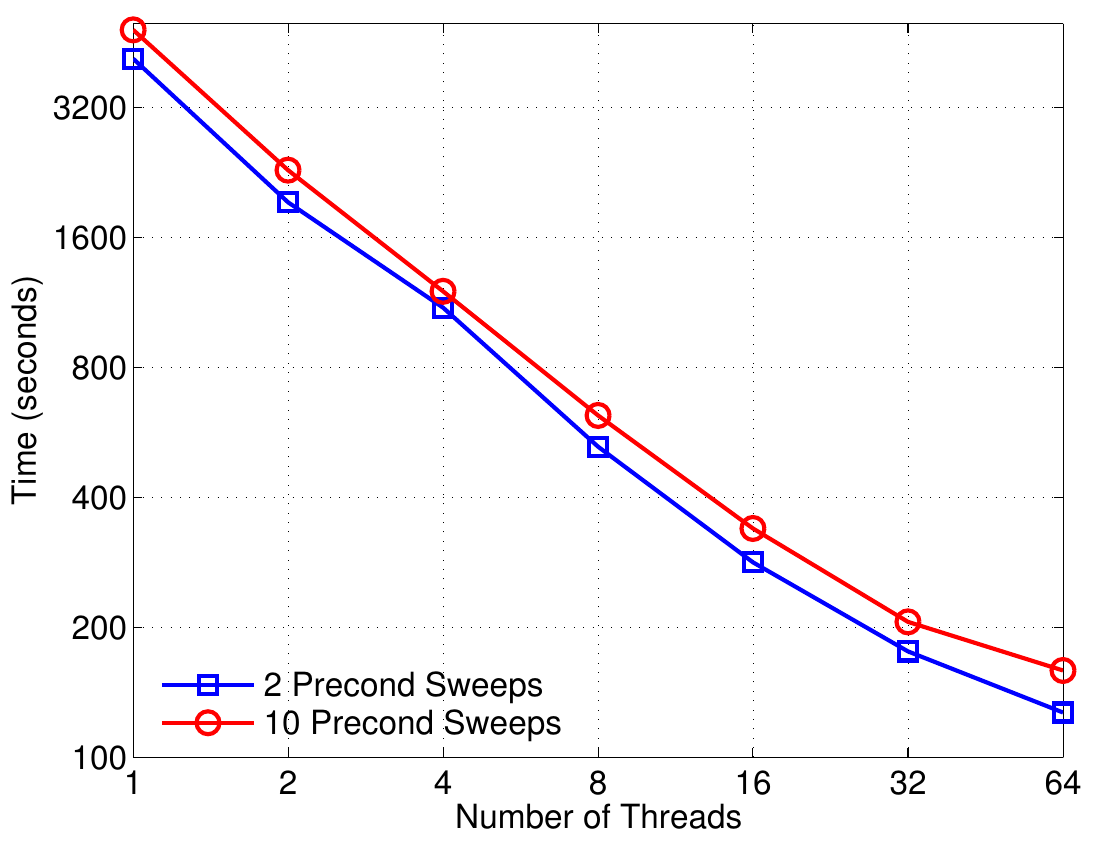} & ~ & \includegraphics[width=0.45\textwidth]{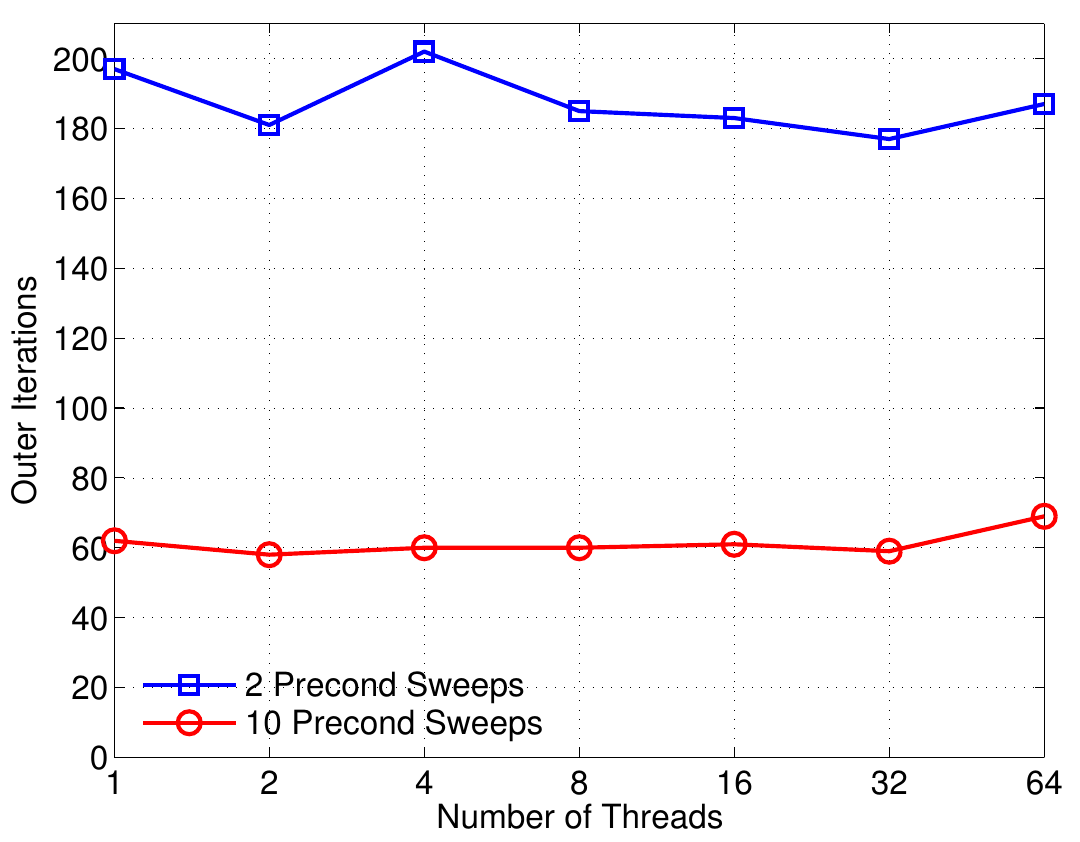}\tabularnewline
\end{tabular}
\par\end{centering}

\protect\caption{\label{fig:args-fcg}Parallel performance of Flexible-CG preconditioned
using AsyRGS (inconsistent read) on the test matrix. Left: Running
time of the algorithm, with 2 or 10 sweeps per preconditioner application,
on the test matrix. Right: Number of outer (Flexible-CG) iterations. }
\end{figure}

There is a trade-off in setting the number of AsyRGS sweeps to use
as a preconditiner. As we increase the number of preconditioner sweeps,
the preconditioner improves so we expect the number of outer (Flexible-CG)
iterations to decrease. On the other hand, AsyRGS is not as effective
on utilizing residual information as CG, which enjoys a superior converge
rate compared to Randomized Gauss-Seidel, so we expect the number
of times the matrix is operated on, which is equal to OuterIterations
$\times$ (InnerSweeps + 1), to increase as we increase the number
of preconditioner sweeps. On the flip side, as we increase the number
of preconditioner sweeps, we expect parallel efficiency to increase
as well since we are diverting work to AsyRGS, which enjoys better
parallel efficiency. Table~\ref{tab:args-fcg} explores this trade-off.
In this experiment we vary the number of inner preconditioner sweeps,
and run Flexible-CG to convergence. $64$ threads are used. As we
expect, the number of outer iterations decreases as the number of
inner sweeps increases, but the overall times the matrix is operated
on increases (with the notable exception of using a single inner sweep).
Improved efficiency is demonstrated in the ``Mat-ops / sec'' columns
(which is equal to OuterIterations $\times$ (InnerSweeps + 1) / Time):
we see that as the number of inner sweeps increase, we operate on
the matrix at a higher rate. Nevertheless, the optimal number of sweeps
(in terms of running time) is obtained with only two inner sweeps.

Parallel performance is explored in Figure~\ref{fig:args-fcg}. In
the left graph we plot the running time as a function of the number
of threads, for two configurations: 2 inner sweeps and 10 inner sweeps.
We see that the method exhibits good speedups, with speedup of more than
$32$ for 2 inner sweeps, and $30$ for 10 inner sweeps. On the surface
it appears that the two inner sweeps configuration enjoys better scalability,
which is counter to our intuition that diverting more work to the
asynchronous iterations should improve parallel efficiency. However,
the running time is measured until convergence, and the number of
(outer) iterations is also a function of the number of threads. This
is explored in the right graph of Figure~\ref{fig:args-fcg}. While
intuitively the number of iterations should grow with the number of
threads, as the quality of the preconditioner should degrade due to
increased asynchronism, that is not observed in practice. We do see
higher variability in the number of iterations with 2 inner sweeps.
Possibly the reason for these observations is that the random choices
made by the algorithm are more dominant than asynchronism in determining
convergence. We now note that the speedup in terms of mat-ops / sec
for 10 inner sweeps is almost $34$, versus only $28$ for 2 inner
sweeps, which is consistent with the intuition that diverting more
work to the asynchronous iterations should improve parallel efficiency.

\section{\label{sec:conclusions}Conclusions and Future Work}

As we push forward toward exascale systems, it is becoming imperative
to revisit asynchronous linear solvers as a means of addressing the
limitations foreseen by current hardware trends. This paper serves
as a starting point for this revisit. Our main observation is that
the limitations of previous asynchronous linear solvers can be addressed
by a new class of asynchronous methods based on randomization. Our
analytical results clearly show the advantage of using randomization
as a building block for asynchronous solvers.

While we do present experimental results that show the usefulness
of our algorithm for certain types of linear systems, it is also clear
that much needs to be done for a general purpose solver. One clear
path, which we only started to explore, is the use of our algorithm
as a preconditioner in a flexible Krylov method. Another, is to extend
our algorithm from a shared memory system with limited parallelism
to massively parallel systems.

There are some theoretical questions that need to be explored too.
Is that gap in the bound for consistent and inconsistent reads inherent,
or an improved analysis can remove or narrow it? Is it possible to
obtain comparable bounds when we allow $k(j)$ or $K(j)$ to depend
on $\d_{0},\dots,\d_{j}$? In our reference scenario, we show weak-scaling
only if we periodically synchronize the threads. It is worth investigating
whether the periodic synchronization is essential, or if it is an artifact
of the analysis. In addition, in our analysis the convergence rate
depends on the maximum age of data used during the algorithm. The
maximum can be rather large in some setups (e.g., high ratio between
maximum and minimum amount of non-zeros per row), but the use of the
maximum is also rather pessimistic (the analysis assumes that the
maximum delay is almost achieved). Perhaps a probabilistic modeling
of the delays might lead to a convergence result that will be
more descriptive for matrices with imbalanced row sizes.

\section*{Acknowledgments}

Thanks to Vikas Sindhwani for providing the matrix used in the experiments.
Haim Avron acknowledges the support from the XDATA program of the
Defense Advanced Research Projects Agency (DARPA), administered through
Air Force Research Laboratory contract FA8750-12-C-0323.

\bibliographystyle{plain}
\bibliography{AvronDruinskyGupta15}

\end{document}